\documentclass[a4paper]{amsart}

\usepackage{a4wide}
\usepackage[UKenglish]{babel}
\usepackage[colorlinks=true,linkcolor=blue,citecolor=red]{hyperref}
\usepackage{mathrsfs}
\usepackage{enumerate}
\usepackage{graphicx}
	\graphicspath{{./}{figures/}}
\usepackage{algorithm,algorithmic}

\theoremstyle{remark}\newtheorem{remark}{Remark}[section]
\theoremstyle{plain}
	\newtheorem{theorem}[remark]{Theorem}
	\newtheorem{proposition}[remark]{Proposition}
	\newtheorem{lemma}[remark]{Lemma}
	\newtheorem{corollary}[remark]{Corollary}

\newcommand{\abs}[1]{\left\vert#1\right\vert}
\newcommand{\cG}{\mathcal{G}}
\newcommand{\cM}{\mathcal{M}}
\newcommand{\cP}{\mathcal{P}}
\newcommand{\thetaFEM}{\bar{\theta}_\textup{FEM}}
\newcommand{\thetaMC}{\bar{\theta}_\textup{MC}}
\newcommand{\norm}[2]{\left\Vert#1\right\Vert_{#2}}
\newcommand{\R}{\mathbb{R}}
\newcommand{\scrC}{\mathscr{C}}
\newcommand{\Z}{\mathbb{Z}}

\begin{document}

\title[Collision avoidance by sidestepping]{Kinetic description of collision avoidance in pedestrian crowds by sidestepping}

\author{Adriano Festa}
\address{RICAM, Austrian Academy of Sciences (\"OAW), Altenbergerstr. 69, 4040 Linz, Austria}
\email{a.festa@ricam.oeaw.ac.at}

\author{Andrea Tosin}
\address{Department of Mathematical Sciences ``G. L. Lagrange'', Politecnico di Torino, Corso Duca degli Abruzzi 24, 10129 Turin, Italy}
\email{andrea.tosin@polito.it}

\author{Marie-Therese Wolfram}
\address{Mathematical Institute, University of Warwick, CV4 7AL Coventry, UK and RICAM, Austrian Academy of Sciences (\"OAW), Altenbergerstr. 69, 4040 Linz, Austria}
\email{m.wolfram@warwick.ac.uk}



\begin{abstract}
In this paper we study a kinetic model for pedestrians, who are assumed to adapt their motion towards a desired direction while avoiding collisions with others by stepping aside. These minimal microscopic interaction rules lead to complex emergent macroscopic phenomena, such as velocity alignment in unidirectional flows and lane or stripe formation in bidirectional flows. We start by discussing collision avoidance mechanisms at the microscopic scale, then we study the corresponding Boltzmann-type kinetic description and its hydrodynamic mean-field approximation in the grazing collision limit. In the spatially homogeneous case we prove directional alignment under specific conditions on the sidestepping rules for both the collisional and the mean-field model. In the spatially inhomogeneous case we illustrate, by means of various numerical experiments, the rich dynamics that the proposed model is able to reproduce.
\end{abstract}

\maketitle

\section{Introduction}
The complex dynamical behaviour of large pedestrian crowds has always fascinated researchers from various scientific fields.  Academic studies began in earnest in the last century, starting with empirical observations in the early 1950’s and continuing with the development of models in the field of applied physics. In recent years, applied mathematicians have become increasingly interested in analytic aspects as well as computational challenges related to simulation and calibration. Modelling the intricate individual dynamics and understanding the emergence of complex macroscopic phenomena, such as the formation of directional lanes or collective patterns, is an active area of research nowadays.

Different classes of models have been proposed in the literature -- either at the micro-, meso- or macroscopic level. More recently also multiscale descriptions to model the impact of single individuals on the dynamics of a larger crowd have been discussed. Microscopic dynamics are based on the change of the individual position and possibly velocity due to interactions with others walkers and with the surrounding. The social force model~\cite{helbing1995PRE}, which is based on Newton's laws of motion, is among the most prominent and successful models in this class. But also cellular automata, see for example~\cite{burstedde2001PHYSA}, or stochastic optimal control approaches, such as the one in~\cite{hoogendoorn2004}, have been studied to model the individual behaviour. In kinetic models, the evolution of the crowd distribution with respect to the microscopic position and velocity of the pedestrians is described by a Boltzmann-type equation, in which interactions are included in the so-called ``collision kernel''; see for example~\cite{degond2013KRM}. Macroscopic PDEs describe instead the crowd as a continuum with density, of which they study the evolution in space and time, see for example \cite{borsche2015JNS,colombo2012M3AS,twarogowska2014AMM}. More recently multiscale approaches, which allow one to model the interactions of single individuals, for example leaders, with a large crowd have been proposed in~\cite{albi2016SIAP,colombi2016JCSMD,cristiani2011MMS}. For a detailed review of mathematical models of pedestrian dynamics we refer to the work by Cristiani and co-workers~\cite[Chapter 4]{cristiani2014BOOK}.

Experiments confirm that collision avoidance is one of the main driving forces in pedestrian dynamics, see~\cite{karamouzas2014PRL}. Individuals actively anticipate the behaviour of others and try to avoid collisions while maintaining their desired direction. The deviations from the desired direction result from stepping aside -- either to the right or the left. In a force based model, such as the social force model proposed by Helbing and co-workers, sidestepping can be included by a force which depends on the estimated collision time and acts perpendicular to the vector connecting the two involved individuals. In cellular automata approaches sidestepping in bidirectional flows can be accounted for by enhancing the transversal transition rates (with respect to the desired direction). Based on these different microscopic collision avoidance mechanisms, various meso- and macroscopic models have been proposed and studied in the literature. At the mesoscopic level these interactions often result in complex collisional operators, see for example the work of Degond and co-workers~\cite{degond2013KRM}. Burger et al.~\cite{burger2016SIMA} studied a minimal macroscopic model for bidirectional flows with sidestepping, which resulted in the segregation of the two flows into separate lanes.

In the present work we propose instead a minimal kinetic description, which is based on the assumptions that individuals try to avoid collisions while moving in their desired direction, see also~\cite{festa2015CDC}. We would like to mention that the collision avoidance plays an important role also in kinetic models of vehicular traffic. For detailed information we refer to~\cite{klar2000BOOKCH,piccoli2009ENCYCLOPEDIA}.

\medskip

This paper is organised as follows: in Section~\ref{sect:micro.coll} we discuss microscopic collision mechanisms. In Section~\ref{sect:kin_mod} we formulate a Boltzmann-type kinetic model for collision avoidance based on the minimal concepts of desired direction and sidestepping. In Section~\ref{sect:spat.homog} we identify conditions for directional alignment in the case of spatially homogeneous pedestrians for both the Boltzmann-type model and its hydrodynamic counterpart under the grazing collision limit. We also show various numerical experiments which confirm our theoretical findings. Finally, in Section~\ref{sect:spat.inhom} we illustrate the dynamics of the proposed model for more general problems under spatially inhomogeneous interaction rules.

\section{Microscopic collision mechanisms}
\label{sect:micro.coll}
We start by discussing various mechanisms of pedestrian collision at the microscopic scale. We consider a crowd of $N>1$ pedestrians, whose position and velocity are denoted by $x_i\in\R^2$ and $v_i\in\R^2$, $i = 1,\,\dots,\,N$, respectively. The dynamics of the pedestrians are driven by two goals: moving with a desired velocity or in a desired direction and, at the same time, trying to avoid collisions with others.

\begin{figure}[!t]
\begin{center}
\includegraphics[width=0.4\textwidth]{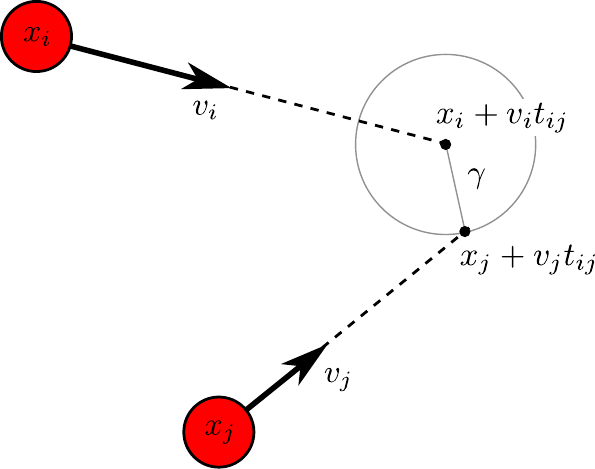}
\caption{Calculation of the time to collision.}
\label{fig:collision}
\end{center}
\end{figure}

We assume that the reaction of a pedestrian to nearby walkers depends on his/her prediction of the time before a possible collision~\cite{karamouzas2014PRL}. This time, the so-called \emph{time to collision}, can be estimated by extrapolating the pedestrian positions on the basis of the current velocities. Specifically, we imagine that the individuals continue to walk with their current velocity, hence along straight paths, until they collide. Furthermore, since we regard pedestrians as point particles, we assume that a collision between two of them, say $i$ and $j$, occurs once they get closer than a certain threshold $\gamma>0$ and we fix the time to collision $t=t_{ij}$ at the instant in which their distance is exactly $\gamma$, cf. Figure~\ref{fig:collision}. Therefore we obtain $t_{ij}$ by solving the equation
\begin{equation}
	\abs{x_i+tv_i-(x_j+tv_j)}^2=\gamma^2
	\label{eq:time.coll}
\end{equation}
with respect to $t$ with the constraints $t\in\R$, $t\geq 0$, which yields
\begin{equation*}
	t_{ij}=t(x_i,\,x_j,\,v_i,\,v_j):=
	\begin{cases}
		0 & \text{if } \abs{x_i-x_j}\leq\gamma \\
		-\dfrac{(x_i-x_j)\cdot(v_i-v_j)+\sqrt{\Delta}}{\abs{v_i-v_j}^2} & \text{if } \abs{x_i-x_j}>\gamma \text{ and~\eqref{eq:time.coll} has solutions} \\
		+\infty & \text{if } \abs{x_i-x_j}>\gamma \text{ and~\eqref{eq:time.coll} has no solution},
	\end{cases}
\end{equation*}
where $\cdot$ denotes the inner product in $\R^2$ while the discriminant $\Delta$ is
$$ \Delta:={[(x_i-x_j)\cdot(v_i-v_j)]}^2-\abs{v_i-v_j}^2\left(\abs{x_i-x_j}^2-\gamma^2\right). $$
Notice that in the case $\abs{x_i-x_j}>\gamma$ a future collision can occur, i.e.~\eqref{eq:time.coll} admits non-negative real solutions, only if the states $(x_i,\,v_i)$, $(x_j,\,v_j)$ of the interacting individuals satisfy
\begin{equation}
	\abs{x_i-x_j}^2-\frac{{[(x_i-x_j)\cdot (v_i-v_j)]}^2}{\abs{v_i-v_j}^2}\leq\gamma^2, \qquad (x_i-x_j)\cdot(v_i-v_j)<0.
	\label{collision}
\end{equation}
While the first condition guarantees $\Delta\geq 0$, hence that the solutions to~\eqref{eq:time.coll} are real, the second condition ensures that the solutions to~\eqref{eq:time.coll} are non-negative. In this case, the smaller one is retained as time to collision.

\begin{figure}[!t]
\begin{center}
\includegraphics[width=0.4\textwidth]{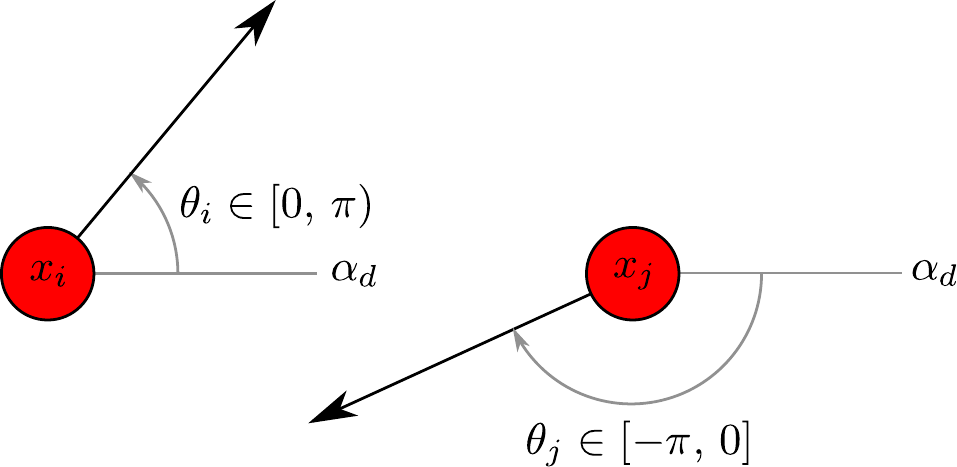}
\caption{The state variables describing pedestrians moving with constant speed with $\alpha_d=0$.}
\label{fig:points}
\end{center}
\end{figure}

In order to reduce the analytical and computational complexity of the model (considering that for two-dimensional position and velocity the dimension of the state space is $4$), it may be convenient to assume that all the pedestrians walk at a constant speed:
$$ \abs{v_i}=V_0>0 \qquad \forall\,i=1,\,\dots,\,N. $$
In particular, working with dimensionless variables, we set $V_0=1$ with reference to the standard walking speed of a pedestrian in normal conditions, i.e. $O(1~\text{m/s})$. This assumption implies:
$$ v_i=V_0(\cos{\theta_i}\mathbf{i}+\sin{\theta_i}\mathbf{j}), $$
where $\mathbf{i}$, $\mathbf{j}$ are the unit vectors of the coordinate axes in $\R^2$ and $\theta_i\in\R$ is the angle giving the orientation of the velocity $v_i$. Hence the microscopic state of the $i$th pedestrian is fully characterised by the pair $(x_i,\,\theta_i)\in\R^3$, see Figure~\ref{fig:points}, which reduces the dimension of the state space by one, cf.~\cite{agnelli2015M3AS}.

\begin{remark} \label{rem:alpha_d}
Because of the $2\pi$-periodicity of the orientation, the angle $\theta_i$ can actually be taken in any bounded interval $I\subset\R$ of length $2\pi$. Throughout this paper we set
$$ I:=[-\pi+\alpha_d,\,\pi+\alpha_d). $$
The interval $I$ is centred around a preferential angle $\alpha_d\in [-\pi,\,\pi)$, which corresponds to the \emph{desired direction}  of the pedestrians, cf. again Figure~\ref{fig:points}.
\end{remark}

In this reduced setting, the relevant expression of the time to collision, i.e. the one which applies when $\abs{x_i-x_j}>\gamma$ and~\eqref{eq:time.coll} has solutions, takes the form:
\begin{equation}
	t_{ij}=t(x_i,\,x_j,\,\theta_i,\,\theta_j):=
		-\dfrac{(x_i-x_j)\cdot\left[(\cos{\theta_i}-\cos{\theta_j})\mathbf{i}+(\sin{\theta_i}-\sin{\theta_j})\mathbf{j}\right]
			+\sqrt{\Delta}}{2V_0\left(1-\cos{(\theta_i+\theta_j)}\right)},
	\label{eq:tij.polar}
\end{equation}
where the discriminant is now given by
\begin{align}
	\begin{aligned}[t]
		\Delta &:= {\left\{(x_i-x_j)\cdot\left[(\cos{\theta_i}-\cos{\theta_j})\mathbf{i}+(\sin{\theta_i}-\sin{\theta_j})\mathbf{j}\right]\right\}}^2 \\
		&\phantom{:=} -2(1-\cos{(\theta_i+\theta_j)})\left(\abs{x_i-x_j}^2-\gamma^2\right).
	\end{aligned}
	\label{eq:Delta.polar}
\end{align}
	
In view of a statistical description of the ensemble of pedestrians, the time to collision can be used to define a \emph{probability of collision} $P$ between two individuals with states $(x_i,\,\theta_i)$ and $(x_j,\,\theta_j)$. For instance:
\begin{equation}
	P(x_i,\,x_j,\,\theta_i,\,\theta_j):=e^{-t_{ij}/\tau},
	\label{eq:prob.coll}
\end{equation}
where $\tau>0$ is a constant. Notice that $0<P\leq 1$ for $t_{ij}\geq 0$ as expected and furthermore $P\to 1$ when $t_{ij}\to 0^+$ while $P\to 0$ when $t_{ij}\to +\infty$. Often, however, the qualitative asymptotic trends of the collision dynamics are studied in \emph{spatially homogeneous conditions}, i.e. when pedestrians are so well mixed that their statistical distribution can be considered independent of the space variable. In this case, the function~\eqref{eq:prob.coll} has to be approximated by a function of the angles $\theta_i$, $\theta_j$ only. In this paper we consider:
\begin{equation}
	P(\theta_i,\,\theta_j)\propto\frac{1}{\pi}\min\{\abs{\theta_i-\theta_j},\,2\pi-\abs{\theta_i-\theta_j}\},
	\label{eq:prob.coll.spathomog}
\end{equation}
i.e. we basically assume that $P$ is proportional to the distance between the angles, which is linked to the incidence of the walking directions, taking into account the $2\pi$-periodicity of the orientation of the velocities in $\R^2$.

\begin{figure}[!t]
\begin{center}
\includegraphics[width=\textwidth]{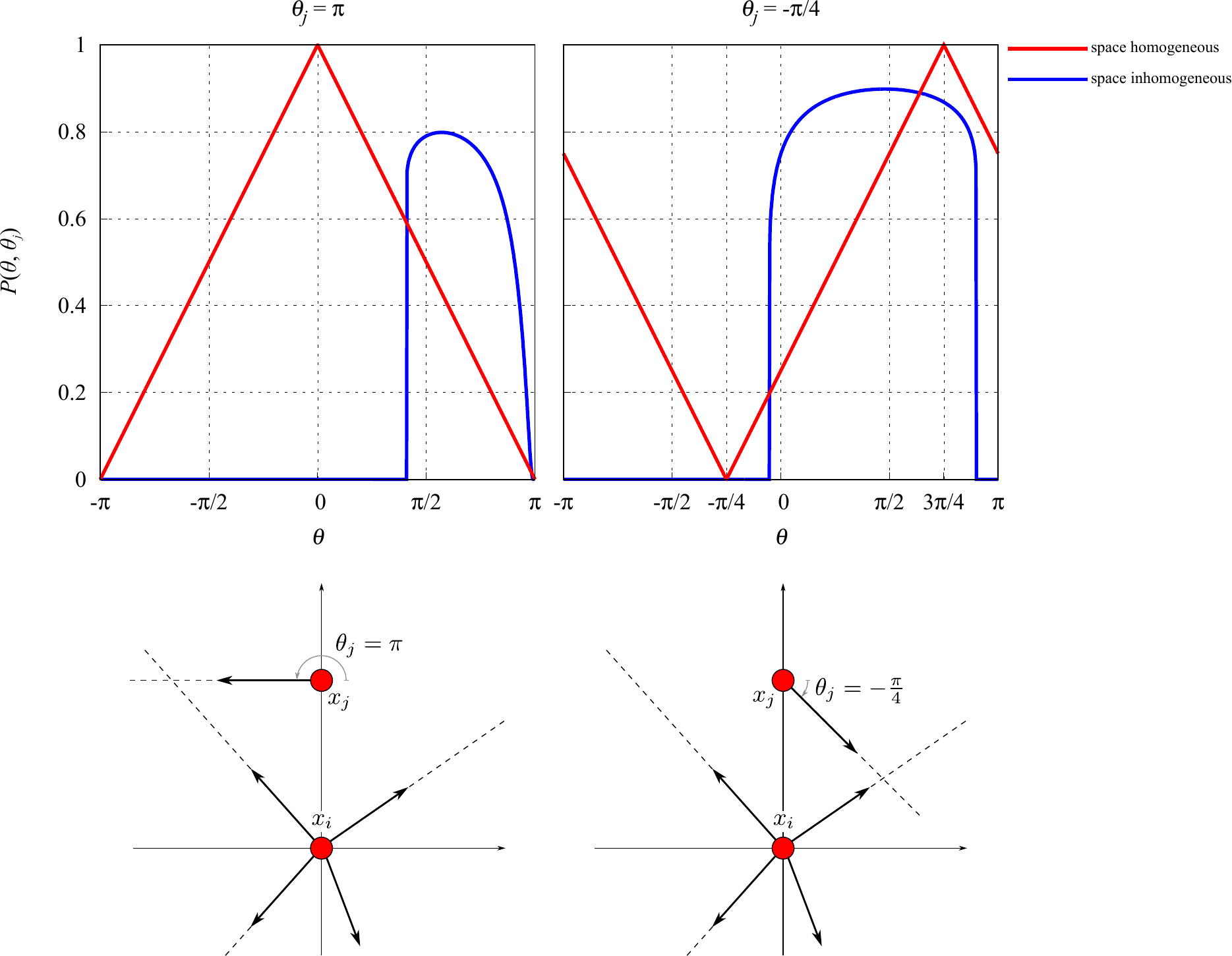}
\caption{Comparison between~\eqref{eq:prob.coll} (blue curves computed with $\gamma=0.8$, $\tau=1$) and~\eqref{eq:prob.coll.spathomog} (red curves). Left: $\theta_j=\pi$; right: $\theta_j=-\frac{\pi}{4}$ (see the text for details).}
\label{prob}
\end{center}
\end{figure}

In Figure~\ref{prob} we compare the functions~\eqref{eq:prob.coll},~\eqref{eq:prob.coll.spathomog} for two pedestrians located in $x_i=(0,\,0)$ and in $x_j=(0,\,1)$ with $\theta_i=\theta$ variable in the range $[-\pi,\,\pi)$ and either $\theta_j=\pi$ or $\theta_j=-\frac{\pi}{4}$. In both cases,~\eqref{eq:prob.coll.spathomog} can be seen as a piecewise linear approximation of~\eqref{eq:prob.coll} in the range of $\theta$ where conditions~\eqref{collision} are met. If instead $\theta$ is such that either $t_{ij}=0$ or $t_{ij}=+\infty$ then the spatially inhomogeneous and homogeneous probabilities of collision may differ consistently, because the spatially homogeneous model does not consider the relative position of the interacting pedestrians.

\section{Kinetic modelling}
\label{sect:kin_mod}
Based on the microscopic modelling discussed before we now frame some principles of collision avoidance in a kinetic context. This approach allows us to make the modelling more amenable to mathematical analysis and to the understanding of group-wise dynamics.

As already discussed, we assume that the pedestrians walk at a constant speed, hence we describe their microscopic state by means of their position in space $x\in\R^2$ and the orientation $\theta\in I$ of their velocity. We denote by $f=f(t,\,x,\,\theta):\R_+\times\R^2\times I\to\R_+$ the kinetic distribution function, which is such that $f(t,\,x,\,\theta)\,dx\,d\theta$ is the fraction of pedestrians who at time $t$ are in the infinitesimal volume $dx$ centred in $x$ with an orientation in $[\theta,\,\theta+d\theta]$. We assume that $f$ obeys the following Boltzmann-type equation:
\begin{equation}
	\partial_t f(t,\,x,\,\theta)+v\cdot\nabla_xf(t,\,x,\,\theta)=Q(f,\,f)(t,\,x,\,\theta),
	\label{eq:Boltz}
\end{equation}
where the collisional term at the right-hand side describes binary interactions between pairs of pedestrians. Explicitly:
\begin{equation}
	Q(f,\,f)(t,\,x,\,\theta)=\int_I\int_{\R^2}\left(\frac{1}{\abs{J}}f(t,\,x,\,\theta_\ast)f(t,\,y,\,\varphi_\ast)
		-f(t,\,x,\,\theta)f(t,\,y,\,\varphi)\right)\,dy\,d\varphi,
	\label{eq:Q}
\end{equation}
with $(x,\,\theta_\ast),\,(y,\,\varphi_\ast)$ and $(x,\,\theta),\,(y,\,\varphi)$ representing the pre- and post-interaction states, respectively, and $\abs{J}$ the determinant of the Jacobian $J$ of the transformation from the former to the latter. Notice that binary interactions are assumed to modify only the orientation of the velocity but not the position of the interacting individuals, which instead changes because of the transport term at the left-hand side of~\eqref{eq:Boltz}.

In order to avoid the computation of the Jacobian $J$ in $Q$ it is customary to use the weak formulation of~\eqref{eq:Boltz}, which is formally obtained by multiplying the equation by a sufficiently smooth and $x$-compactly supported test function $\psi:\R^2\times I\to\R$ and integrating with over $x,$ and $\theta$:
\begin{multline}
	\frac{d}{dt}\int_I\int_{\R^2}\psi(x,\,\theta)f(t,\,x,\,\theta)\,dx\,d\theta
		+\int_I\int_{\R^2}v\cdot\nabla_x\psi(x,\,\theta)f(t,\,x,\,\theta)\,dx \,d\theta= \\
	\int_I\int_I\int_{\R^2}\int_{\R^2}\left(\psi(x,\,\theta)
		-\psi(x,\theta_\ast)\right)f(t,\,x,\,\theta_\ast)f(t,\,y,\,\varphi_\ast)\,dx\,dy\,d\theta_\ast\,d\varphi_\ast.
	\label{eq:Boltz.weak}
\end{multline}
This form is easier to handle and can provide information about the evolution of some macroscopic quantities of the system.

To characterise the collisional term we propose an interaction rule of the form
\begin{equation*}
	\begin{cases}
		\theta=\scrC(x,\,y,\,\theta_\ast,\,\varphi_\ast)+2h\pi \\
		\varphi=\scrC(y,\,x,\,\varphi_\ast,\,\theta_\ast)+2k\pi,
	\end{cases}
	\quad h,\,k\in\Z
\end{equation*}
with
\begin{equation}
	\scrC(x,\,y,\,\theta_\ast,\,\varphi_\ast):=\theta_\ast+(1-P(x,\,y,\,\theta_\ast,\,\varphi_\ast))(\alpha_d-\theta_\ast)
		+P(x,\,y,\,\theta_\ast,\,\varphi_\ast)\alpha_c,
	\label{eq:int.rules}
\end{equation}
where $h,\,k\in\Z$ are chosen to ensure that $\theta,\,\varphi\in I$. Therefore they depend in general on $\theta_\ast,\,\varphi_\ast\in I$. In~\eqref{eq:int.rules}, $\alpha_d$ is the desired angle introduced in Remark~\ref{rem:alpha_d} whereas the \emph{deviation angle} $\alpha_c\in [-\pi,\,\pi)$ corresponds to the lateral displacement of the pedestrians, when they step aside to avoid collisions. A pedestrian steps to the left when  $\alpha_c>0$ and to the right when $\alpha_c<0$ (with respect to their current direction of motion). Finally, the term $P(x,\,y,\,\theta_\ast,\,\varphi_\ast)$ is the probability of collision introduced in Section~\ref{sect:micro.coll}.

\section{The spatially homogeneous problem}
\label{sect:spat.homog}
In the space homogeneous case pedestrians are assumed to be well mixed so that their distribution function $f$ is constant in $x$, i.e. $f=f(t,\,\theta)$. In other words, the statistical distribution of the angle $\theta$ is the same at every point $x$.

The weak formulation~\eqref{eq:Boltz.weak} of the kinetic equation then becomes:
\begin{equation}
	\frac{d}{dt}\int_I\psi(\theta)f(t,\,\theta)\,d\theta=
		\int_I\int_I\left(\psi(\theta)-\psi(\theta_\ast)\right)f(t,\,\theta_\ast)
			f(t,\,\varphi_\ast)\,d\theta_\ast d\varphi_\ast.
	\label{eq:Boltz.weak.spat_homog}
\end{equation}
We point out that, in contrast to~\eqref{eq:Boltz.weak}, here the test function $\psi:I\to\R$ does not have to be either smooth or compactly supported. Furthermore, in the following we assume that the test functions are extended to $\mathbb{R}$ by $2\pi$-periodicity, so as to avoid unneccessary technicalities caused  the periodic boundary conditions. In particular, taking the constant test function $\psi\equiv 1$ we see that the \emph{macroscopic density}
$$ \rho(t):=\int_If(t,\,\theta)\,d\theta $$
is constant in time because the right-hand side of~\eqref{eq:Boltz.weak.spat_homog} vanishes. This means that the crowd density can be regarded as a parameter, say $\rho$, of the model fixed by the initial condition.

By dropping the transport term $v\cdot\nabla_x{f}$ in the equation, the assumption of space homogeneity allows one to study the genuine effect of the microscopic interactions, which are now expressed as
$$ \theta=\scrC(\theta_\ast,\,\varphi_\ast)+2h\pi, \quad h\in\Z $$
with
\begin{equation}
	\scrC(\theta_\ast,\,\varphi_\ast)=\theta_\ast+(1-P(\theta_\ast,\,\varphi_\ast))(\alpha_d-\theta_\ast)
		+P(\theta_\ast,\,\varphi_\ast)\alpha_c
	\label{eq:int.rules.spathomog}
\end{equation}
and $P$ given by~\eqref{eq:prob.coll.spathomog}. In more detail, we assume
\begin{equation}
	P(\theta_\ast,\,\varphi_\ast):=a(\rho)\cG(\abs{\theta_\ast-\varphi_\ast}),
		\qquad \cG(s):=\frac{1}{\pi}\min\{s,\,2\pi-s\}
	\label{eq:P.space_homog}
\end{equation}
where $\cG:[0,\,2\pi)\to [0,\,1]$ corresponds to a function which belongs to ~\eqref{eq:prob.coll.spathomog}, while $a(\rho)\in [0,\,1]$ models the influence of the collective state of the crowd on the individual interaction rules. Considering a dimensionless $\rho\in [0,\,1]$ referred to a typical congestion density $O(4~\text{ped/m\textsuperscript{2}})$, cf.~\cite{polus1983JTE,weidmann1992TECHREP}, we suggest:
\begin{enumerate}[(i)]
\item $a(\rho)\propto\rho$ to model a probability of collision proportional to the number of pedestrians in the crowd;
\item $a(\rho)\propto\rho(1-\rho)$ to account for a small probability of collision in the case of spread out groups ($\rho\approx 0$) and in very crowded situations ($\rho\approx 1$), when pedestrians tend to be passively dragged by the flow. The highest probability of collision arises at intermediate congestion levels ($\rho\approx\frac{1}{2}$).
\end{enumerate}

\subsection{Asymptotic alignment under binary interactions}
\label{sect:asympt.align_binary}
Now we investigate the conditions under which interactions lead pedestrians to move in the same direction. We call this \emph{alignment}, which is a form of \emph{consensus}~\cite{albi2014PTRSA,canuto2012SICON,motsch2014SIREV}. We say that there is alignment when the kinetic distribution $f$ converges asymptotically in time to a distribution of the form $\rho\delta_{\alpha}$ for some $\alpha\in I$, where $\delta_\alpha$ is the Dirac delta centred at $\theta=\alpha$. In fact this means that, in the long run, the mass concentrates at $\theta=\alpha$.

We shall use the following notations and definitions:
\begin{itemize}
\item Let $\cM^\rho_+(I)$ denote the set of the positive measures on $I$ with total mass $\rho\geq 0$;
\item Let $W_1(\mu,\,\nu)$ denote the \emph{$1$-Wasserstein distance}~\cite{ambrosio2008BOOK} for $\mu,\,\nu\in\cM^\rho_+(I)$, given by
$$ W_1(\mu,\,\nu):=\inf_{\lambda\in\Pi(\mu,\,\nu)}\iint_{I^2}\vert\theta-\varphi\vert\,d\lambda(\theta,\,\varphi), $$
where $\Pi(\mu,\,\nu)$ is the set of the transference plans from $\mu$ to $\nu$, i.e. the measures on $I^2$ with marginals $\mu$ and $\nu$.
\end{itemize}

First of all, we show that $\alpha_d$ is the only possible direction of alignment.
\begin{proposition} \label{prop:only_alphad.Boltz}
A distribution of the form $\rho\delta_\alpha\in\cM^\rho_+(I)$, with $\rho>0$ and $\alpha\in I$, is a steady solution to~\eqref{eq:Boltz.weak.spat_homog} with interaction rule~\eqref{eq:int.rules.spathomog} and interaction probability~\eqref{eq:P.space_homog} if and only if $\alpha=\alpha_d$.
\end{proposition}
\begin{proof}
Plugging $\rho\delta_\alpha$ into~\eqref{eq:Boltz.weak.spat_homog} and invoking the $2\pi$-periodicity of $\psi$ gives
$$ \psi(\scrC(\alpha,\,\alpha))-\psi(\alpha)=0, $$
that is, using~\eqref{eq:int.rules.spathomog} and considering from~\eqref{eq:P.space_homog} that $P(\alpha,\,\alpha)=0$,
$$ \psi(\alpha_d)-\psi(\alpha)=0. $$
This equation has to hold for all test functions $\psi$, therefore $\alpha=\alpha_d$.
\end{proof}

Next we study under which conditions the system aligns in the desired direction $\alpha_d$ starting with a possibly generic initial distribution $f_0\in\cM^\rho_+(I)$.

\begin{theorem} \label{theo:align.Boltz}
Let $\rho>0$ and $P$ be given by~\eqref{eq:P.space_homog}. Define moreover
$$ \bar{\theta}_0:=\frac{1}{\rho}\int_I\abs{\theta-\alpha_d}f_0(\theta)\,d\theta, $$
where $f_0\in\cM^\rho_+(I)$ is a prescribed initial distribution function.

If
\begin{equation}
	\bar{\theta}_0<\pi\left(\frac{1}{a(\rho)}-1\right) \quad \text{and} \quad
		\abs{\alpha_c}\leq\frac{\pi}{2}\left(\frac{1}{a(\rho)}-1\right)-\frac{\bar{\theta}_0}{2}
	\label{eq:cond.align.Boltz}
\end{equation}
then 
$$ \lim_{t\to+\infty}W_1(f(t),\,\rho\delta_{\alpha_d})=0. $$
\end{theorem}
\begin{proof}
We define for $\rho>0$ and for every $t>0$ the average distance of the crowd from the desired direction 
\begin{equation}
	\bar{\theta}(t):=\frac{1}{\rho}\int_I\abs{\theta-\alpha_d}f(t,\,\theta)\,d\theta.
	\label{eq:bartheta}
\end{equation}
Proving that $\lim\limits_{t\to+\infty}\bar{\theta}(t)=0$ implies the thesis, because $f(t,\,\theta)\otimes\rho\delta_{\alpha_d}(\varphi)$ is a possible transference plan in $\Pi(f(t),\,\rho\delta_{\alpha_d})$ and therefore
\begin{align*}
	W_1(f(t),\,\rho\delta_{\alpha_d}) &\leq \iint_{I^2}\vert\theta-\varphi\vert\,d(f(t,\,\theta)\otimes\rho\delta_{\alpha_d}(\varphi)) \\
	& =\rho\int_I\abs{\theta-\alpha_d}f(t,\,\theta)\,d\theta=\rho^2\bar{\theta}(t)\xrightarrow{t\to+\infty}0.
\end{align*}

Choosing $\psi(\theta)=\frac{1}{\rho}\abs{\theta-\alpha_d}$ as test function in~\eqref{eq:Boltz.weak.spat_homog} yields
\begin{equation}
	\frac{d}{dt}\bar{\theta}(t)=\frac{1}{\rho}\int_I\int_I\left(\abs{\theta-\alpha_d}
		-\abs{\theta_\ast-\alpha_d}\right)f(t,\,\theta_\ast)f(t,\,\varphi_\ast)\,d\theta_\ast\,d\varphi_\ast.
	\label{eq:weak.proof.onepop}
\end{equation}
In particular, from the interaction rules~\eqref{eq:int.rules} we have 
$$ \abs{\theta-\alpha_d}=\abs{\scrC(\theta_\ast,\,\varphi_\ast)+2h\pi-\alpha_d}\leq\abs{\scrC(\theta_\ast,\,\varphi_\ast)-\alpha_d}. $$
In fact the value of $h\in\Z$ is chosen such that $\scrC(\theta_\ast,\,\varphi_\ast)+2h\pi\in I$, whereas in general one may have $\scrC(\theta_\ast,\,\varphi_\ast)\not\in I$. In view of this and using the interaction probability~\eqref{eq:P.space_homog} together with the triangular inequality we obtain
\begin{align*}
	\abs{\theta-\alpha_d}-\abs{\theta_\ast-\alpha_d} &\leq (P-1)\abs{\theta_\ast-\alpha_d}+P\abs{\alpha_c} \\
	&= \bigl[a(\rho)\cG(\abs{\theta_\ast-\varphi_\ast})-1\bigr]\abs{\theta_\ast-\alpha_d}
		+a(\rho)\cG(\abs{\theta_\ast-\varphi_\ast})\abs{\alpha_c} \\
	&\leq \left[\frac{a(\rho)}{\pi}\abs{\theta_\ast-\varphi_\ast}-1\right]\abs{\theta_\ast-\alpha_d}
		+\frac{a(\rho)}{\pi}\abs{\alpha_c}\abs{\theta_\ast-\varphi_\ast}
\intertext{where we have used that $\cG(\abs{\theta_\ast-\varphi_\ast})\leq\frac{1}{\pi}\abs{\theta_\ast-\varphi_\ast}$, see \eqref{eq:P.space_homog}. Moreover we have that $\abs{\theta_\ast-\varphi_\ast}\leq\abs{\theta_\ast-\alpha_d}+\abs{\varphi_\ast-\alpha_d}$, thus}
	\abs{\theta-\alpha_d}-\abs{\theta_\ast-\alpha_d}&\leq \frac{a(\rho)}{\pi}\abs{\theta_\ast-\alpha_d}^2+\frac{a(\rho)}{\pi}\abs{\theta_\ast-\alpha_d}\cdot\abs{\varphi_\ast-\alpha_d}
		-\abs{\theta_\ast-\alpha_d} \\
	&\phantom{\leq} +\frac{a(\rho)}{\pi}\abs{\alpha_c}\left(\abs{\theta_\ast-\alpha_d}+\abs{\varphi_\ast-\alpha_d}\right).
\intertext{Finally, $\abs{\theta_\ast-\alpha_d}\leq\pi$ implies $\abs{\theta_\ast-\alpha_d}^2\leq\pi\abs{\theta_\ast-\alpha_d}$, hence}
	\abs{\theta-\alpha_d}-\abs{\theta_\ast-\alpha_d}&\leq a(\rho)\abs{\theta_\ast-\alpha_d}+\frac{a(\rho)}{\pi}\abs{\theta_\ast-\alpha_d}\cdot\abs{\varphi_\ast-\alpha_d}
		-\abs{\theta_\ast-\alpha_d} \\
	&\phantom{\leq} +\frac{a(\rho)}{\pi}\abs{\alpha_c}\left(\abs{\theta_\ast-\alpha_d}+\abs{\varphi_\ast-\alpha_d}\right).
\end{align*}

Plugging this into~\eqref{eq:weak.proof.onepop} and using the definition of $\bar{\theta}$ we get
$$ \frac{d\bar{\theta}}{dt}\leq\lambda\rho\bar{\theta}+\mu\rho\bar{\theta}^2, $$
where
\begin{equation}
	\lambda:=a(\rho)\left(1+\frac{2\abs{\alpha_c}}{\pi}\right)-1, \qquad
		\mu:=\frac{a(\rho)}{\pi}.
	\label{eq:lambda.mu}
\end{equation}
In order to integrate the previous differential inequality we define $w(t):=\bar{\theta}^{-1}(t)$, whence
$$ \frac{dw}{dt}\geq -\lambda\rho w-\mu\rho, $$
which yields $w(t)\geq\bar{\theta}_0^{-1}e^{-\lambda\rho t}+\frac{\mu}{\lambda}\left(e^{-\lambda\rho t}-1\right)$ for $\bar{\theta}_0=\bar{\theta}(0)=w^{-1}(0)$. Going back to $\bar{\theta}$ we deduce that
\begin{equation}
	\bar{\theta}(t)\left[\frac{e^{-\lambda\rho t}}{\bar{\theta}_0}+\frac{\mu}{\lambda}\left(e^{-\lambda\rho t}-1\right)\right]\leq 1,
		\quad \forall\,t\geq 0.
	\label{eq:proof.theta.convergence}
\end{equation}
	
From~\eqref{eq:cond.align.Boltz} we have
$$ \lambda\leq-\mu\bar{\theta}_0 \quad \text{and} \quad
	\frac{e^{-\lambda\rho t}}{\bar{\theta}_0}+\frac{\mu}{\lambda}\left(e^{-\lambda\rho t}-1\right)\geq\frac{1}{\bar{\theta}_0}, $$
consequently, if $a(\rho),\,\bar{\theta}_0>0$,
$$ \lim_{t\to+\infty}\left[\frac{e^{-\lambda\rho t}}{\bar{\theta}_0}+\frac{\mu}{\lambda}\left(e^{-\lambda\rho t}-1\right)\right]=+\infty, $$
which, in view of~\eqref{eq:proof.theta.convergence}, implies that $\bar{\theta}(t)\to 0$ for $t\to+\infty$.

In order to complete the proof we examine the cases $a(\rho)=0$, $\bar{\theta}_0=0$. In the former we have $\mu=0$ and $\lambda=-1$, thus from~\eqref{eq:proof.theta.convergence} we deduce $\bar{\theta}(t)\leq\bar{\theta}_0e^{-\rho t}\to 0$ for $t\to+\infty$. In the latter we deduce instead $\bar{\theta}(t)=0$, i.e. $f(t,\,\theta)=\rho\delta_{\alpha_d}(\theta)$, for all $t\geq 0$, which concludes the proof.
\end{proof}
\begin{remark}
From the proof of Theorem~\ref{theo:align.Boltz}, if the conditions~\eqref{eq:cond.align.Boltz} are met we infer the following estimate for $\bar{\theta}$:
\begin{equation}
	\bar{\theta}(t)\leq \frac{\abs{\lambda}\bar{\theta}_0}{\left(\abs{\lambda}-\mu\bar{\theta}_0\right)e^{\abs{\lambda}\rho t}+\mu\bar{\theta}_0},
		\quad \forall\,t\geq 0
	\label{eq:estimate.bartheta}
\end{equation}
where $\lambda,\,\mu$ are given in~\eqref{eq:lambda.mu}. In Section~\ref{sect:spat.homog.Boltz_numsim} we compare this estimate with the numerical values of $\bar{\theta}(t)$ computed by simulating the kinetic model with a Monte Carlo algorithm.
\end{remark}

\subsection{Numerical simulations}
\label{sect:spat.homog.Boltz_numsim}
Next we confirm the theoretical results presented in Section~\ref{sect:asympt.align_binary} with numerical simulations. We choose $a(\rho)=\rho$ and  $\alpha_d=0$ in the following, therefore $I=[-\pi,\,\pi)$. We solve~\eqref{eq:Boltz.weak.spat_homog} by the Nanbu-like algorithm, which is an implementation of the Monte Carlo (MC) method for kinetic equations (see Appendix~\ref{app:nanbu} and~\cite[Chapter 4]{pareschi2013BOOK}), performing $4$ runs per test with $N=5\cdot 10^5$ particles in each run.

\begin{figure}[!t]
\begin{center}
\includegraphics[width=0.4\textwidth]{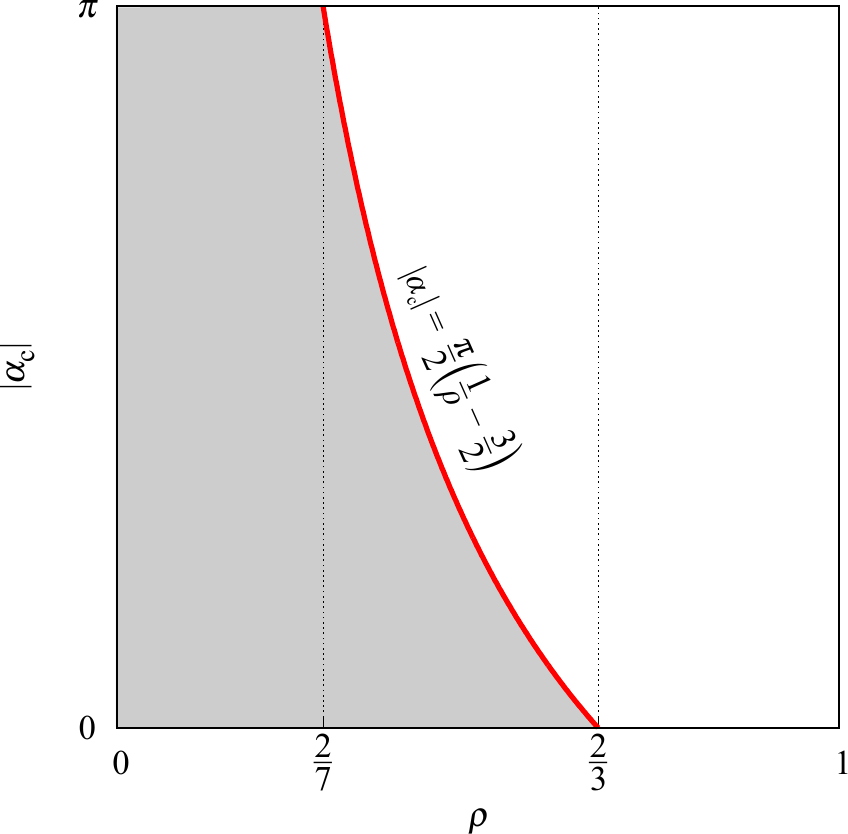}
\caption{The set of pairs $(\rho,\,\abs{\alpha_c})\in [0,\,1]\times [0,\,\pi]$ (shaded area) for which Theorem~\ref{theo:align.Boltz} guarantees alignment in the case study of Section~\ref{sect:spat.homog.Boltz_numsim}.}
\label{fig:alpha_c.binary}
\end{center}
\end{figure}

As initial condition we choose the uniform distribution over all possible directions with mass $\rho$:
\begin{equation}
	f_0(\theta)=\frac{\rho}{2\pi}, \quad \theta\in I,
	\label{eq:f0.uniform}
\end{equation}
which implies $\bar{\theta}_0=\frac{\pi}{2}$. Conditions~\eqref{eq:cond.align.Boltz} of Theorem~\ref{theo:align.Boltz} read then
\begin{equation}
	0\leq\rho<\frac{2}{3}, \qquad \abs{\alpha_c}\leq\frac{\pi}{2}\left(\frac{1}{\rho}-\frac{3}{2}\right).
	\label{eq:cond.align.Boltz.case_study}
\end{equation}
Figure~\ref{fig:alpha_c.binary} illustrates the regions of alignment in this particular case.

\begin{figure}[!t]
\begin{center}
\includegraphics[width=0.8\textwidth]{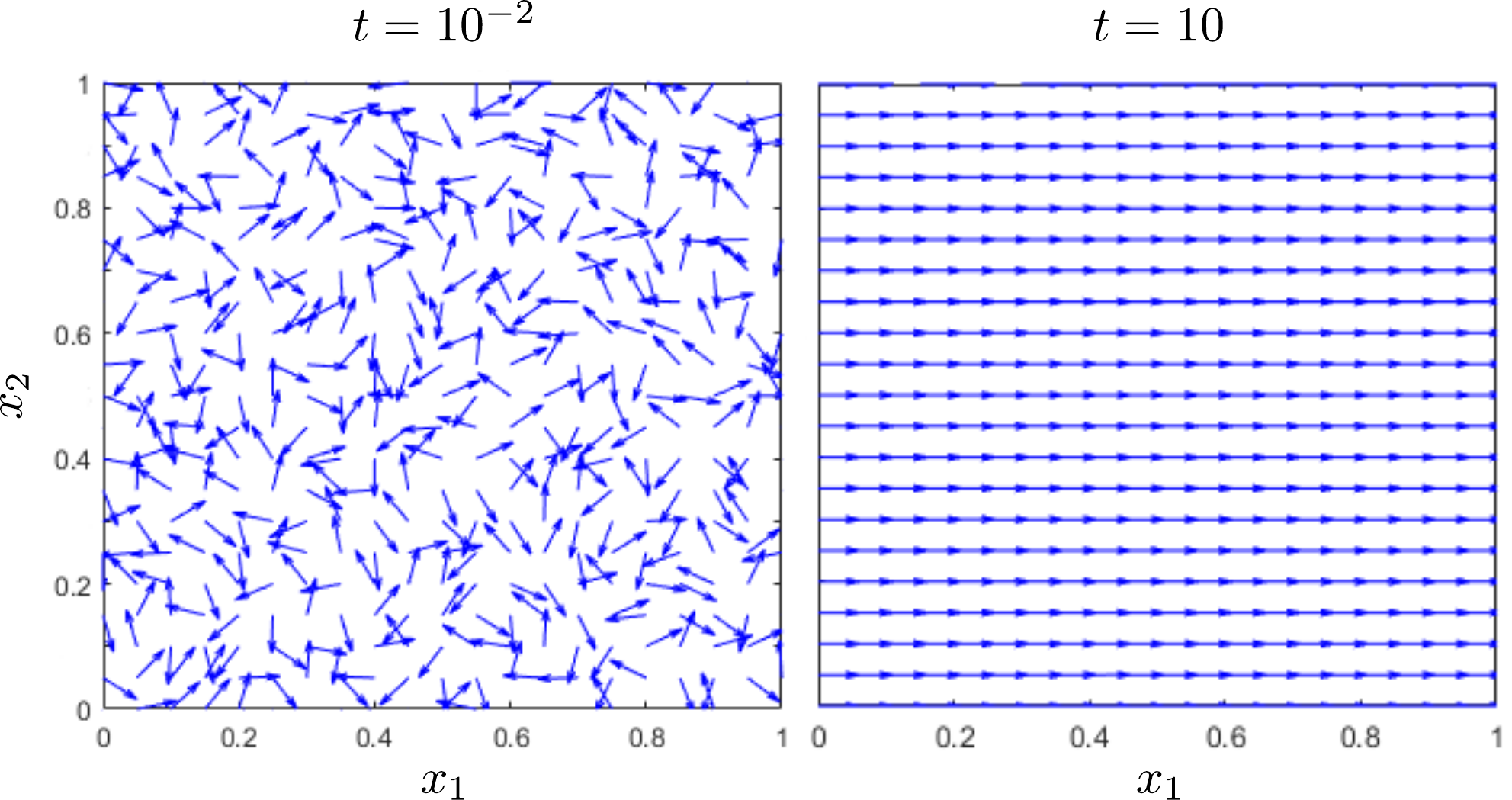}
\caption{Alignment to the desired direction $\alpha_d=0$ for the case study of Section~\ref{sect:spat.homog.Boltz_numsim} with $\rho=\frac{1}{2}$, $\alpha_c=\frac{\pi}{5}$. Left: $t=10^{-2}$, right: $t=10$.}
\label{fig:align.spat_homog.Boltz}
\end{center}
\end{figure}

In Figure~\ref{fig:align.spat_homog.Boltz}, we fix $\rho=\frac{1}{2}$ and $\alpha_c=\frac{\pi}{5}$, which comply with~\eqref{eq:cond.align.Boltz.case_study}, and we represent the mean velocity of the pedestrians at two successive times, namely $t=10^{-2}$ after one iteration of the algorithm and $t=10$ after $10^3$ iterations. Note that Figure \ref{fig:align.spat_homog.Boltz} shows only the subdomain $[0,\,1]\times [0,\,1]$, to illustrate the alignment at $\alpha_d = 0$ better.

\begin{figure}[!t]
\begin{center}
\includegraphics[width=0.7\textwidth]{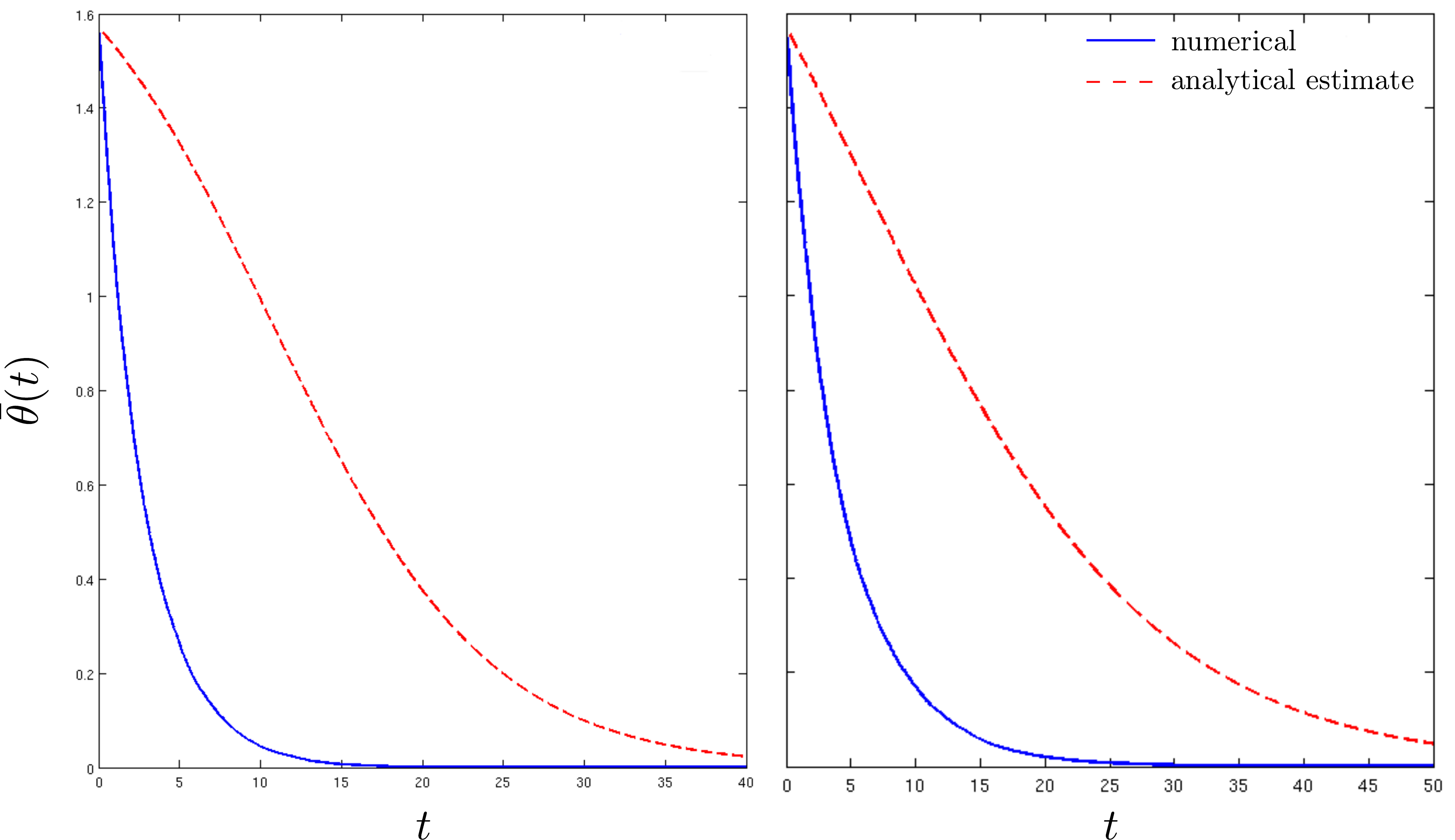}
\caption{The mapping $t\mapsto\bar{\theta}(t)$ computed numerically (blue solid curve) and estimated theoretically by~\eqref{eq:estimate.bartheta} (red dashed curve). Left: $(\rho,\,\alpha_c)=\left(\frac{1}{2},\,\frac{\pi}{5}\right)$; right: $(\rho,\,\alpha_c)=\left(\frac{1}{3},\,\frac{3}{5}\pi\right)$.}
\label{fig:thetabar}
\end{center}
\end{figure}

Next we analyse the rate of convergence in time to the desired direction by plotting an MC approximation, say $\thetaMC$, of $\bar{\theta}$ cf.~\eqref{eq:bartheta}. Following~\cite[Chapter 3]{pareschi2013BOOK} we take
$$ \thetaMC(t):=\frac{1}{N}\sum^N_{i=1}\abs{\Theta_i(t)-\alpha_d}, $$
where $N$ is the number of particles used in the MC simulations and the $\Theta_i$'s are $N$ values of the angle $\theta$ sampled from the (MC-approximated) probability distribution $\frac{1}{\rho}f(t,\,\theta)$. In Figure~\ref{fig:thetabar} we compare the mapping $t\mapsto\thetaMC(t)$ with the theoretical estimate~\eqref{eq:estimate.bartheta} for the choices $(\rho,\,\alpha_c)=\left(\frac{1}{2},\,\frac{\pi}{5}\right)$ on the left and $(\rho,\,\alpha_c)=\left(\frac{1}{3},\,\frac{3}{5}\pi\right)$ on the right (notice that also the latter pair satisfies~\eqref{eq:cond.align.Boltz.case_study}). In both examples the simulations converge faster than the theoretical rate.
\begin{figure}[!t]
\begin{center}
\includegraphics[width=\textwidth]{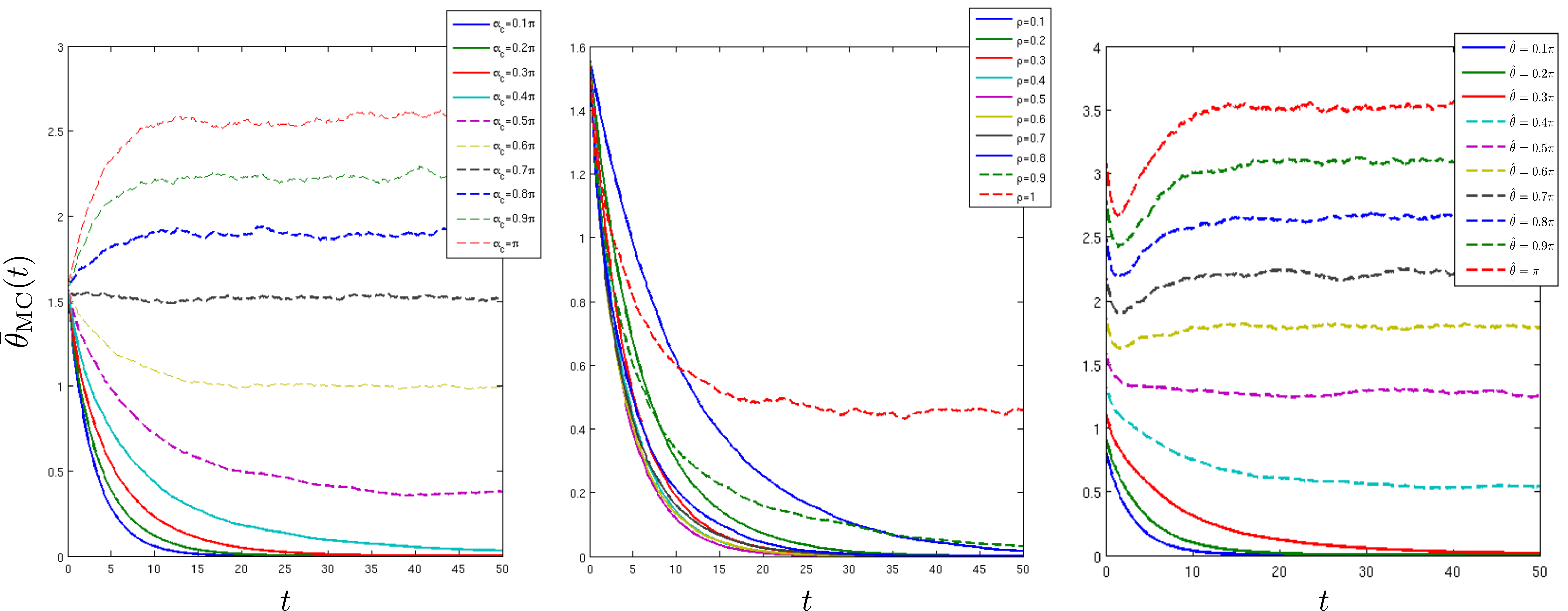}
\caption{The mapping $t\mapsto\thetaMC(t)$ for: (\textbf{left}) $\rho=\frac{1}{2}$ and different choices of $\alpha_c\in (0,\,\pi]$; (\textbf{centre}) $\alpha_c=\frac{\pi}{5}$ and different choices of $\rho\in (0,\,1]$; (\textbf{right}) $\rho=1$, $\alpha_c=\frac{2}{5}\pi$ and different values of $\bar{\theta}_0$ obtained by varying the parameter $\hat{\theta}$ in~\eqref{eq:f0.Gaussian}. Dashed curves indicate the cases in which alignment fails.}
\label{fig:thetabarMC.parameters}
\end{center}
\end{figure}

In Figure~\ref{fig:thetabarMC.parameters} we further investigate the dependence of the asymptotic alignment on $\rho$, $\alpha_c$ and $\bar{\theta}_0$ by plotting the evolution of $\thetaMC$ in the case:
\begin{enumerate}[(i)]
\item $\rho=\frac{1}{2}$ and setting $\alpha_c=\frac{k}{10}\pi$ for $k=1,\,\dots,\,10$, cf. Figure~\ref{fig:thetabarMC.parameters} (left);
\item $\alpha_c=\frac{\pi}{5}$ and setting $\rho=\frac{k}{10}$ for $k=1,\,\dots,\,10$, cf. Figure~\ref{fig:thetabarMC.parameters} (centre);
\item $\rho=1$, $\alpha_c=\frac{2}{5}\pi$ and taking as initial condition the distribution
\begin{equation}
	f_0(\theta)=\frac{\rho}{\sqrt{2\pi\sigma^2}}\sum_{n=-\infty}^{+\infty}e^{-\frac{{\left(\theta+2n\pi-\hat{\theta}\right)}^2}{2\sigma^2}},
		\qquad \theta\in I,
	\label{eq:f0.Gaussian}
\end{equation}
i.e. the Gaussian bell with mean $\hat{\theta}\in I$ and variance $\sigma^2>0$ ``folded'' by $2\pi$-periodicity of $\theta$ into the interval $I$. It can be checked that such an $f_0$ has mass $\rho$ for all $\hat{\theta}$, $\sigma^2$, while the corresponding value of $\bar{\theta}_0$ varies with $\hat{\theta}$, $\sigma^2$. In particular, we let $\sigma^2=1$ and $\hat{\theta}=\frac{k}{10}\pi$ for $k=1,\,\dots,\,10$, cf. Figure~\ref{fig:thetabarMC.parameters} (right).
\end{enumerate}

In case (i),~\eqref{eq:cond.align.Boltz.case_study} predicts alignment for $\abs{\alpha_c}\leq\frac{\pi}{4}$ but the numerical simulations show alignment also for $\alpha_c=\frac{3}{10}\pi,\,\frac{2}{5}\pi$. Note that Theorem~\ref{theo:align.Boltz} provides only \emph{sufficient} conditions to characterise the asymptotic trend of the system. On the other hand, the numerical simulations confirm that alignment is not possible for every value of $\alpha_c$. Specifically, we observe that $\thetaMC$ always converges to some possibly non-zero asymptotic value, nonetheless, due to Proposition~\ref{prop:only_alphad.Boltz}, this does not indicate a possible alignment in a direction different from $\alpha_d$. It means instead that interactions may, in some cases, preserve or even increase the initial ``disorder'' of the pedestrian orientations. In fact we see that for $\alpha_c\geq\frac{7}{10}\pi$ the value of $\thetaMC$ either remains almost constant or increases in time.

In case (ii),~\eqref{eq:cond.align.Boltz.case_study} predicts alignment for $\rho\leq\frac{10}{19}\approx 0.52$ but the numerical simulations show alignment also for $\rho < \frac{9}{10}$.

In case (iii) the conditions~\eqref{eq:cond.align.Boltz} of Theorem~\ref{theo:align.Boltz} are always violated, because for $\rho=1$ it results $a(\rho)=1$ and subsequently $\bar{\theta}_0<0$. In spite of this, the numerical simulations show that in some cases the alignment is possible, particularly when $\hat{\theta}$ is sufficiently close to $\alpha_d=0$. This can be understood by noticing that for $\hat{\theta}=\alpha_d=0$ the mean of~\eqref{eq:f0.Gaussian} is precisely $\hat{\theta}$, for $f_0$ is symmetric in $I$. Therefore on average the pedestrians are initially not too far from their desired direction. However, such an asymptotic alignment is soon lost by varying $\hat{\theta}$, hence $\bar{\theta}_0$.

The tests (i)--(iii) confirm that, although the hypotheses of Theorem~\ref{theo:align.Boltz} are possibly overly-restrictive, the theoretical predictions capture qualitatively the role played by the parameters of the model in the alignment process.

\subsection{The quasi-invariant direction limit}
In this section we study the \emph{quasi-invariant direction limit}, namely a form of grazing collision limit~\cite{villani2002BOOKCH}, see also~\cite{pareschi2013BOOK,toscani2006CMS}, in order to extract the principal part of the interaction rules encoded in the Boltzmann collisional operator. This procedure allows us to derive a hydrodynamic-type model closer to the macroscopic scale, in which binary interactions are replaced by mean-field interactions defining a transport velocity in the space of the microscopic states.

The basic idea is to assume that pedestrian interactions get simultaneously weaker and more frequent. To this purpose we begin by rewriting~\eqref{eq:int.rules.spathomog} as
\begin{equation}
	\theta=\theta_\ast+\epsilon[(1-P(\theta_\ast,\,\varphi_\ast))(\alpha_d-\theta_\ast)
		+P(\theta_\ast,\,\varphi_\ast)\alpha_c]+2h\pi,
	\label{eq:int.rules.spathomog.eps}
\end{equation}
where $\epsilon>0$ is a dimensionless parameter measuring the strength of the interaction. For $\epsilon=1$ we recover precisely~\eqref{eq:int.rules.spathomog}. At the same time, we scale the time in the Boltzmann equation~\eqref{eq:Boltz.weak.spat_homog} as $t/\epsilon$, meaning that the interaction rate becomes $O(1/\epsilon)$, so that we finally have
\begin{equation}
	\frac{d}{dt}\int_I\psi(\theta)f(t,\,\theta)\,d\theta=
		\frac{1}{\epsilon}\int_I\int_I\left(\psi(\theta)-\psi(\theta_\ast)\right)f(t,\,\theta_\ast)
			f(t,\,\varphi_\ast)\,d\theta_\ast d\varphi_\ast,
	\label{eq:Boltz.weak.spat_homog.eps}
\end{equation}
where now we take $\psi\in C^2(I)$ with $\psi(\pm\pi+\alpha_d)=0$, still $2\pi$-periodic on the whole real line.

As we are interested in small values of $\epsilon$, in view of~\eqref{eq:int.rules.spathomog.eps} and of the said periodicity of $\psi$ we can expand
\begin{align*}
	\psi(\theta)-\psi(\theta_\ast) &= \epsilon\psi'(\theta_\ast)[\alpha_d-\theta_\ast+P(\theta_\ast,\,\varphi_\ast)(\alpha_c-\alpha_d+\theta_\ast)] \\
	&\phantom{=} +\frac{\epsilon^2}{2}\psi''(\hat{\theta})[\alpha_d-\theta_\ast+P(\theta_\ast,\,\varphi_\ast)(\alpha_c-\alpha_d+\theta_\ast)]^2
\end{align*}
for some $\hat{\theta}\in I$ depending on $\theta_\ast$ and $\theta$. Let us define
$$ \cP(\theta_\ast,\,\varphi_\ast):=\alpha_d-\theta_\ast+P(\theta_\ast,\,\varphi_\ast)(\alpha_c-\alpha_d+\theta_\ast) $$
for brevity; plugging into~\eqref{eq:Boltz.weak.spat_homog.eps} we obtain
$$ \frac{d}{dt}\int_I\psi(\theta)f(t,\,\theta)\,d\theta=\int_I\psi'(\theta_\ast)\left(\int_I\cP(\theta_\ast,\,\varphi_\ast)
	f(t,\,\varphi_\ast)\,d\varphi_\ast\right)f(t,\,\theta_\ast)\,d\theta_\ast+R(\epsilon), $$
where, considering that $\abs{\cP(\theta_\ast,\,\varphi_\ast)}\leq\pi$ because both angles $\alpha_d-\theta_\ast$ and $\alpha_c$ belong to $[-\pi,\,\pi)$ while $0\leq P(\theta_\ast,\,\varphi_\ast)\leq 1$, the remainder $R(\epsilon)$ satisfies
$$ \abs{R(\epsilon)}=\frac{\epsilon}{2}\abs{\int_I\int_I\psi''(\hat{\theta})\cP^2(\theta_\ast,\,\varphi_\ast)
	f(t,\,\theta_\ast)f(t,\,\varphi_\ast)\,d\theta_\ast\,d\varphi_\ast}\leq\frac{\epsilon}{2}\norm{\psi''}{\infty}\pi^2\rho^2, $$
with $\norm{\cdot}{\infty}$ denoting the usual $\infty$-norm. Therefore in the limit $\epsilon\to 0^+$ we find
\begin{equation}
	\frac{d}{dt}\int_I\psi(\theta)f(t,\,\theta)\,d\theta=\int_I\psi'(\theta_\ast)H[f](\theta_\ast)f(t,\,\theta_\ast)\,d\theta_\ast,
	\label{eq:quasinv_limit.weak}
\end{equation}
where, recalling also~\eqref{eq:P.space_homog}, we define
\begin{align}
	\begin{aligned}[t]
		H[f](\theta_\ast) &:= \int_I\cP(\theta_\ast,\,\varphi_\ast)f(t,\,\varphi_\ast)\,d\varphi_\ast \\
		&\phantom{:}= \rho(\alpha_d-\theta_\ast)+a(\rho)(\alpha_c-\alpha_d+\theta_\ast)\int_I\cG(\abs{\theta_\ast-\varphi_\ast})f(t,\,\varphi_\ast)\,d\varphi_\ast.
	\end{aligned}
	\label{eq:H[f]}
\end{align}
Equation~\eqref{eq:quasinv_limit.weak} corresponds to the weak form of the PDE
\begin{equation}
	\partial_tf+\partial_\theta(H[f]f)=0,
	\label{eq:quasinv_limit.strong}
\end{equation}
which is a non-linear conservation law with non-local flux.
	
\subsection{Asymptotic alignment under mean-field interactions}
Like for the Boltzmann-type model, we are interested in the steady state solutions of~\eqref{eq:quasinv_limit.strong} to find out under which conditions this equation predicts the emergence of group-wise alignment. Since~\eqref{eq:quasinv_limit.strong} has been derived as an approximation of~\eqref{eq:Boltz.weak.spat_homog} on a time scale much larger than that of the binary interactions (recall the scaling $t\to t/\epsilon$), it is reasonable to expect that the large-time behaviour of its solutions, including the conditions for alignment, approximates the asymptotic trends of~\eqref{eq:Boltz.weak.spat_homog}.

To begin with, by a straightforward calculation we show, like in Proposition~\ref{prop:only_alphad.Boltz}, that also for~\eqref{eq:quasinv_limit.strong} the only possible alignment corresponds to $\alpha_d$.
\begin{proposition}
A distribution of the form $\rho\delta_\alpha\in\cM^\rho_+(I)$, with $\rho>0$ and $\alpha\in I$, is a weak steady solution to~\eqref{eq:quasinv_limit.strong} with transport speed~\eqref{eq:H[f]} if and only if $\alpha=\alpha_d$.
\end{proposition}
\begin{proof}
Substituting $\rho\delta_\alpha$ in~\eqref{eq:quasinv_limit.weak} gives
$$ \rho\psi'(\alpha)H[\rho\delta_\alpha](\alpha)=0. $$
Since $\cG(0)=0$ we further have $H[\rho\delta_\alpha](\alpha)=\rho(\alpha_d-\alpha)$, cf.~\eqref{eq:H[f]}. This equation has to hold for every $\psi$, therefore $\alpha = \alpha_d$.
\end{proof}

Next we claim that, at least in a suitable density-dependent range of values of the deviation angle $\alpha_c$, the group-wise alignment at $\theta=\alpha_d$ is the configuration reached asymptotically in time by the mean-field interaction model \emph{regardless of its initial configuration}. In other words, we can give (sufficient) conditions on $\rho$, $\alpha_c$ ensuring that $\rho\delta_{\alpha_d}$ is a \emph{globally attractive equilibrium} to~\eqref{eq:quasinv_limit.strong}.

To tackle this we introduce the \emph{characteristic line} of~\eqref{eq:quasinv_limit.strong} issuing from $\theta\in I$ as the mapping $t\mapsto\gamma_t(\theta)$ satisfying
\begin{equation}
	\begin{cases}
		\dot{\gamma}_t(\theta)=H[f](\gamma_t(\theta)), & t>0 \\[2mm]
		\gamma_0(\theta)=\theta,
	\end{cases}
	\label{eq:characteristics}
\end{equation}
where the dot over a variable means henceforth time derivative. In practice, $\gamma_t(\theta)$ is the direction at time $t>0$ of a pedestrian who initially ($t=0$) moved in the direction $\theta$. We point out that $\gamma_t(\theta)$ is understood ``modulus $2\pi$'', meaning that, as a function of $\theta$, it maps $I$ into itself at every time. Since~\eqref{eq:quasinv_limit.strong} describes a transport at speed $H[f]$ of the initial kinetic distribution $f_0\in\cM^\rho_+(I)$ in the state space $I$, the (weak) solution $f$ can be written formally by pushing $f_0$ forwards in time along the characteristics:
\begin{equation}
	f(t)=\gamma_t\#f_0\in\cM^\rho_+(I), \quad t>0.
	\label{eq:pushfwd}
\end{equation}

We study the characteristics of~\eqref{eq:quasinv_limit.strong} by the following lemma, which establishes the continuous dependence of the mapping $t\mapsto\gamma_t(\theta)$ on the angle $\theta\in I$.

\begin{lemma} \label{lemma:char.cont_dep}
Set $C:=1+\frac{\abs{\alpha_c}}{\pi}$ with $\alpha_c\in [-\pi,\,\pi)$. For all $\theta_1,\,\theta_2\in I$ it holds:
$$ \abs{\gamma_t(\theta_2)-\gamma_t(\theta_1)}\leq\abs{\theta_2-\theta_1}e^{-\rho[1-a(\rho)(1+C)]t}, \quad t\geq 0. $$
\end{lemma}
\begin{proof}
First we notice that, using~\eqref{eq:pushfwd}, we can rewrite the transport speed $H[f]$ as
\begin{align*}
	H[f](\theta) &= H[\gamma_t\#f_0](\theta) \\
	&= \rho(\alpha_d-\theta)+a(\rho)(\alpha_c-\alpha_d+\theta)\int_I\cG(\abs{\theta-\gamma_t(\varphi)})f_0(\varphi)\,d\varphi.
\end{align*}

We fix now $\theta_1,\,\theta_2\in I$. Writing~\eqref{eq:characteristics} first for $\theta=\theta_1$ then for $\theta=\theta_2$ and taking the difference of the two equations we get
\begin{align*}
	\frac{d}{dt}(\gamma_t(\theta_2) &- \gamma_t(\theta_1)) \\
	&= -\left\{\rho-a(\rho)\int_I\cG(\abs{\gamma_t(\theta_2)-\gamma_t(\varphi)})f_0(\varphi)\,d\varphi\right\}
		\left(\gamma_t(\theta_2)-\gamma_t(\theta_1)\right) \\
	&\phantom{=} +a(\rho)\int_I\Bigl[\cG(\abs{\gamma_t(\theta_2)-\gamma_t(\varphi)})
		-\cG(\abs{\gamma_t(\theta_1)-\gamma_t(\varphi)})\Bigr]f_0(\varphi)\,d\varphi \\
	&\phantom{=} \quad\times(\gamma_t(\theta_1)-\alpha_d+\alpha_c).
\end{align*}

Let us temporarily call $\mathcal{I}$ the second term at the right-hand side. Since $\gamma_t(\theta_1)\in I$ for all $t>0$ by $2\pi$-periodicity, it results $\gamma_t(\theta_1)-\alpha_d+\alpha_c\in [-\pi+\alpha_c,\,\pi+\alpha_c)$, hence
$$ \abs{\gamma_t(\theta_1)-\alpha_d+\alpha_c}\leq\max\{\abs{-\pi+\alpha_c},\,\abs{\pi+\alpha_c}\}=\pi+\abs{\alpha_c}=\pi C, $$
where $C$ is the constant defined in the statement of the lemma. Moreover, from~\eqref{eq:P.space_homog} we see that the function $\cG$ is Lipschitz continuous with Lipschitz constant equal to $\frac{1}{\pi}$, therefore on the whole we have
\begin{align*}
	\abs{\mathcal{I}} &\leq Ca(\rho)\int_I\Bigl\vert\abs{\gamma_t(\theta_2)-\gamma_t(\varphi)}-\abs{\gamma_t(\theta_1)-\gamma_t(\varphi)}\Bigr\vert
		f_0(\varphi)\,d\varphi \\
	&\leq Ca(\rho)\int_I\abs{\gamma_t(\theta_2)-\gamma_t(\theta_1)}f_0(\varphi)\,d\varphi \\
	&=Ca(\rho)\rho\abs{\gamma_t(\theta_2)-\gamma_t(\theta_1)}.
\end{align*}
This implies on one hand
\begin{align*}
	\dfrac{d}{dt}(\gamma_t(\theta_2) &- \gamma_t(\theta_1)) \\
	&\leq -\left\{\rho-a(\rho)\int_I\cG(\abs{\gamma_t(\theta_2)-\gamma_t(\varphi)})f_0(\varphi)\,d\varphi\right\}\left(\gamma_t(\theta_2)-\gamma_t(\theta_1)\right) \\
	&\phantom{\leq} +Ca(\rho)\rho\abs{\gamma_t(\theta_2)-\gamma_t(\theta_1)}
\end{align*}
and on the other hand
\begin{align*}
	\dfrac{d}{dt}(\gamma_t(\theta_2) &- \gamma_t(\theta_1)) \\
	&\geq -\left\{\rho-a(\rho)\int_I\cG(\abs{\gamma_t(\theta_2)-\gamma_t(\varphi)})f_0(\varphi)\,d\varphi\right\}\left(\gamma_t(\theta_2)-\gamma_t(\theta_1)\right) \\
	&\phantom{\leq} -Ca(\rho)\rho\abs{\gamma_t(\theta_2)-\gamma_t(\theta_1)},
\end{align*}
thus, setting
\begin{align*}
	\lambda(t) &:= -\left\{\rho-a(\rho)\int_I\cG(\abs{\gamma_t(\theta_2)-\gamma_t(\varphi)})f_0(\varphi)\,d\varphi\right\} \\
	u(t) &:= \gamma_t(\theta_2)-\gamma_t(\theta_1)
\end{align*}
for ease of notation,
$$ -Ca(\rho)\rho\abs{u(t)}\leq\dot{u}(t)-\lambda(t)u(t)\leq Ca(\rho)\rho\abs{u(t)}. $$

Let $\Lambda(t):=\int_0^t\lambda(s)\,ds$. Multiplying the chain of inequalities above by $e^{-\Lambda(t)}$ and integrating in time gives
$$ -Ca(\rho)\rho\int_0^te^{\Lambda(t)-\Lambda(s)}\abs{u(s)}\,ds\leq u(t)-e^{\Lambda(t)}u_0\leq Ca(\rho)\rho\int_0^te^{\Lambda(t)-\Lambda(s)}\abs{u(s)}\,ds, $$
where $u_0:=u(0)=\gamma_0(\theta_2)-\gamma_0(\theta_1)=\theta_2-\theta_1$. Therefore
$$ \abs{u(t)-e^{\Lambda(t)}u_0}\leq Ca(\rho)\rho\int_0^te^{\Lambda(t)-\Lambda(s)}\abs{u(s)}\,ds $$
and subsequently
\begin{equation}
	\abs{u(t)}\leq e^{\Lambda(t)}\abs{u_0}+Ca(\rho)\rho\int_0^te^{\Lambda(t)-\Lambda(s)}\abs{u(s)}\,ds.
	\label{eq:u}
\end{equation}

Since $0\leq\cG\leq 1$ we observe that $\lambda(t)\leq -\rho(1-a(\rho))$, thus $\Lambda(t)-\Lambda(s)\leq -\rho(1-a(\rho))(t-s)$ for all $0\leq s\leq t$. Therefore from~\eqref{eq:u} we deduce
$$ \abs{u(t)}\leq e^{-\rho(1-a(\rho))t}\abs{u_0}+Ca(\rho)\rho\int_0^te^{-\rho(1-a(\rho))(t-s)}\abs{u(s)}\,ds. $$
Invoking now Gronwall's inequality we deduce that
$$ \abs{u(t)}\leq\abs{u_0}e^{-\rho[1-a(\rho)(1+C)]t}, $$
which, after substituting the definitions of $u(t),\,u_0$ in terms of $\gamma_t,\,\gamma_0$, concludes the proof.
\end{proof}

Next we establish a sufficient condition on the deviation angle $\alpha_c$ such that all the characteristics of~\eqref{eq:quasinv_limit.strong}  converge in time to the desired angle $\alpha_d$.

\begin{proposition} \label{prop:char.conv_alphad}
Let $\alpha_c\in [-\pi,\,\pi)$ be such that
$$ \abs{\alpha_c}<\pi\left(\frac{1}{a(\rho)}-2\right). $$
Then
$$ \lim_{t\to+\infty}\abs{\gamma_t(\theta)-\alpha_d}=0, \quad \forall\,\theta\in I. $$
\end{proposition}
\begin{proof}
We preliminarily observe that, under the stated constraint on $\alpha_c$, the constant $C$ introduced in Lemma~\ref{lemma:char.cont_dep} is such that $a(\rho)(1+C)<1$. Moreover, for $\theta,\,\varphi\in I$ it results $\abs{\theta-\varphi}<2\pi$, thus from Lemma~\ref{lemma:char.cont_dep} we get on the whole
$$ \abs{\gamma_t(\theta)-\gamma_t(\varphi)}<2\pi e^{-\rho[1-a(\rho)(1+C)]t}\xrightarrow{t\to+\infty} 0 $$
uniformly with respect to $\theta,\,\varphi$. We can rephrase this by saying that for all $\epsilon>0$ there exists a time $t_0>0$, independent of $\theta,\,\varphi$, such that $\abs{\gamma_t(\theta)-\gamma_t(\varphi)}<\epsilon$ for all $t>t_0$. Since $\cG$ is continuous with $\cG(0)=0$, we conclude that there also exists $t_0'>0$ such that $\cG(\abs{\gamma_t(\theta)-\gamma_t(\varphi)})<\epsilon$ for all $t>t_0'$ and all $\theta,\,\varphi\in I$.

Let us now consider the equation of the characteristics:
$$ \dot{\gamma}_t(\theta)=-\rho(\gamma_t(\theta)-\alpha_d)+a(\rho)\left(\int_I\cG(\abs{\gamma_t(\theta)-\gamma_t(\varphi)})f_0(\varphi)\,d\varphi\right)
	(\gamma_t(\theta)-\alpha_d+\alpha_c). $$
For $t>\bar{t}:=\max\{t_0,\,t_0'\}$ we obtain the following estimate:
$$ \abs{a(\rho)\left(\int_I\cG(\abs{\gamma_t(\theta)-\gamma_t(\varphi)})f_0(\varphi)\,d\varphi\right)(\gamma_t(\theta)-\alpha_d+\alpha_c)}
	<Ca(\rho)\rho\epsilon, $$
which implies
$$ \dot{\gamma}_t(\theta)<-\rho(\gamma_t(\theta)-\alpha_d)+Ca(\rho)\rho\epsilon, \qquad
	\dot{\gamma}_t(\theta)>-\rho(\gamma_t(\theta)-\alpha_d)-Ca(\rho)\rho\epsilon $$
and subsequently
$$ -Ca(\rho)\rho\epsilon<\dot{\gamma}_t(\theta)+\rho(\gamma_t-\alpha_d)<Ca(\rho)\rho\epsilon. $$

We set $u(t):=\gamma_t(\theta)-\alpha_d$ and observe that $\dot{u}(t)=\dot{\gamma}_t(\theta)$, multiply the chain of inequalities above by $e^{\rho t}$ and integrate over the time interval $[\bar{t},\,t]$ to obtain
$$ -Ca(\rho)\rho\epsilon\int_{\bar{t}}^te^{-\rho(t-s)}\,ds<u(t)-e^{-\rho(t-\bar{t})}u(\bar{t})<Ca(\rho)\rho\epsilon\int_{\bar{t}}^te^{-\rho(t-s)}\,ds, $$
that is
\begin{align*}
	\abs{u(t)}-e^{-\rho(t-\bar{t})}\abs{u(\bar{t})} &\leq \abs{u(t)-e^{-\rho(t-\bar{t})}u(\bar{t})}<Ca(\rho)\rho\epsilon\int_{\bar{t}}^te^{-\rho(t-s)}\,ds \\
	&= Ca(\rho)\left(1-e^{-\rho(t-\bar{t})}\right)\epsilon.
\end{align*}
In the limit $t\to+\infty$ we find
$$ \lim_{t\to+\infty}\abs{u(t)}=\lim_{t\to+\infty}\abs{\gamma_t(\theta)-\alpha_d}\leq Ca(\rho)\epsilon, $$
which, since $\epsilon>0$ was arbitrary, implies the thesis.
\end{proof}

Thanks to Proposition~\ref{prop:char.conv_alphad}, the main result of this section is now easy to prove.
\begin{theorem} \label{theo:align.mean_field}
Let $\alpha_c\in [-\pi,\,\pi)$ be such that
\begin{equation}
	\abs{\alpha_c}<\pi\left(\frac{1}{a(\rho)}-2\right).
	\label{eq:cond.align.mean_field}
\end{equation}
Then
$$ \lim_{t\to+\infty}W_1(f(t),\,\rho\delta_{\alpha_d})=0 $$
for all $f_0\in\cM^\rho_+(I)$.
\end{theorem}
\begin{proof}
Proceeding like in Theorem~\ref{theo:align.Boltz} we have
\begin{align*}
	W_1(f(t),\,\rho\delta_{\alpha_d})\leq\iint_{I^2}\abs{\theta-\varphi}\,d(f(t,\,\theta)\otimes\rho\delta_{\alpha_d}(\varphi))
		&= \rho\int_I\abs{\theta-\alpha_d}f(t,\,\theta)\,d\theta \\
	&= \rho\int_I\abs{\gamma_t(\theta)-\alpha_d}f_0(\theta)\,d\theta,
\end{align*}
hence, by dominated convergence and Proposition~\ref{prop:char.conv_alphad},
$$ \lim_{t\to+\infty}W_1(f(t),\,\rho\delta_{\alpha_d})\leq\rho\int_I\lim_{t\to+\infty}\abs{\gamma_t(\theta)-\alpha_d}f_0(\theta)\,d\theta=0, $$
which concludes the argument.
\end{proof}

\begin{remark} \label{rem:comp_conditions}
Conditions~\eqref{eq:cond.align.Boltz} and~\eqref{eq:cond.align.mean_field} resemble each other, thereby confirming the expectation that the asymptotic behaviour of the mean-field interaction equation should approximate that of the binary interaction model. However, a remarkable difference, which makes it not immediate to compare the two conditions in general, is that the former depends on the initial condition $f_0$ through the quantity $\bar{\theta}_0$ while the latter is independent of $f_0$.

We can get a clue on the relationship between the two conditions by fixing the special case $f_0(\theta)=\frac{\rho}{2\pi}$, $\theta\in I$, which corresponds to the initially most ``disordered'' situation with pedestrians following homogeneously all possible directions. Conditions~\eqref{eq:cond.align.Boltz} specialise then in~\eqref{eq:cond.align.Boltz.case_study} and can be compared with~\eqref{eq:cond.align.mean_field} for different choices of the function $a(\rho)$.

\begin{figure}[!t]
\begin{center}
\includegraphics[width=0.8\textwidth]{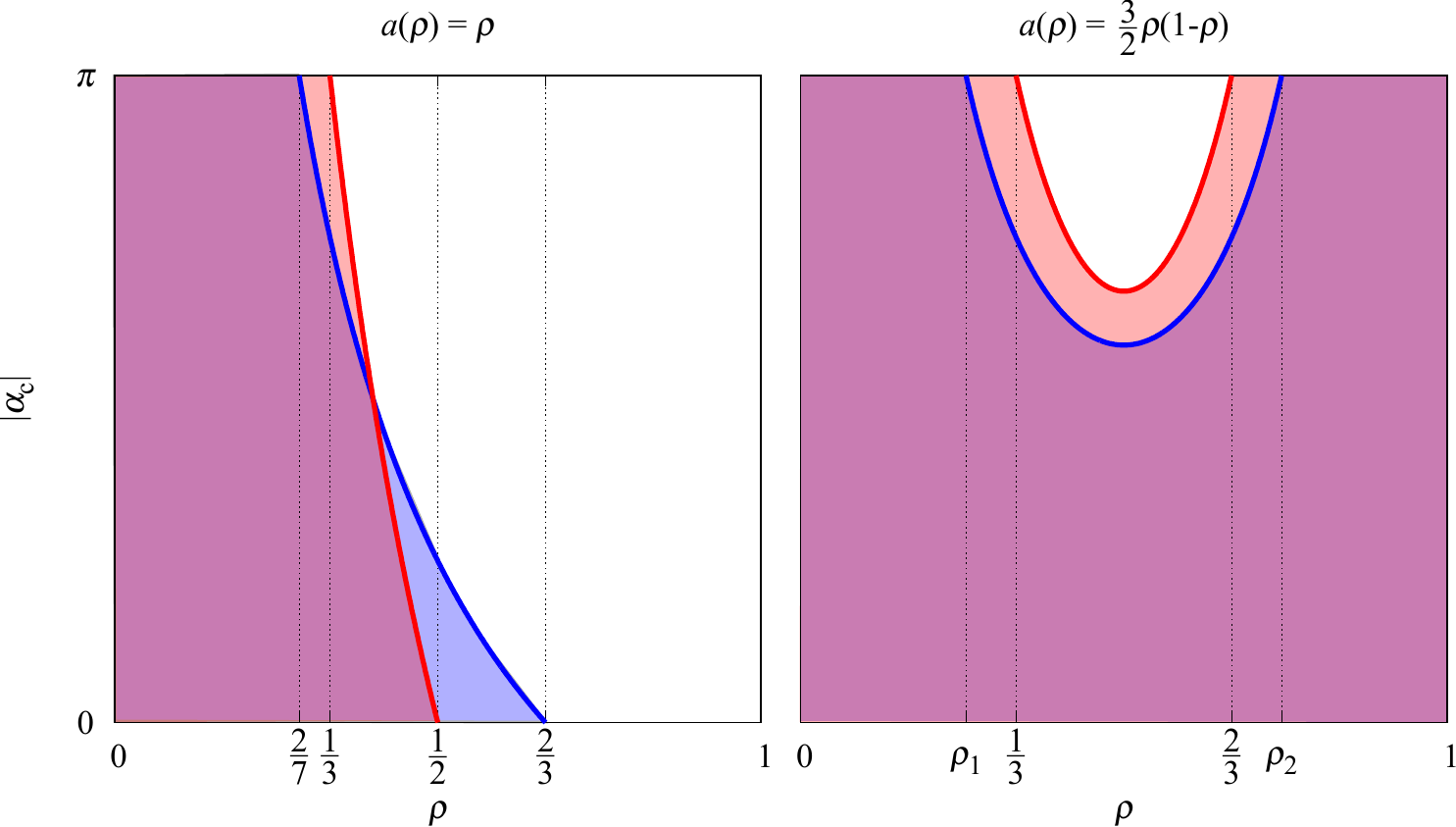}
\caption{Conditions~\eqref{eq:cond.align.Boltz} (blue) with $\bar{\theta}_0=\frac{\pi}{2}$ and~\eqref{eq:cond.align.mean_field} (red) for two different choices of $a(\rho)$. Purple areas correspond to the pairs $(\rho,\,\abs{\alpha_c})$ satisfying both conditions. The values $\rho_1$, $\rho_2$ in the right graph are $\rho_{1,2}=\frac{21\mp\sqrt{105}}{42}$, i.e. $\rho_1\approx 0.256$, $\rho_2\approx 0.744$.}
\label{fig:comp_conditions}
\end{center}
\end{figure}

In Figure~\ref{fig:comp_conditions} we consider the cases $a(\rho)=\rho$ and $a(\rho)=\frac{3}{2}\rho(1-\rho)$ already discussed at the very beginning of Section~\ref{sect:spat.homog}. In the former we see that either condition may be more or less restrictive depending on the considered density range. In the latter we observe instead that the mean-field interaction model may guarantee alignment in a range of values of $\alpha_c$ invariably larger than the corresponding one of the binary interaction model. Finally, in both cases we notice that the density range for \emph{unconditional alignment}, i.e. alignment for every $\alpha_c\in [-\pi,\,\pi)$, is typically larger in the case of the mean-field interaction model.
\end{remark}

We should stress that the considerations proposed in Remark~\ref{rem:comp_conditions} apply \emph{a priori} to the two types of model, but alignment in situations not encompassed by either~\eqref{eq:cond.align.Boltz} or~\eqref{eq:cond.align.mean_field} can still be observed \emph{a posteriori}. In fact these conditions are on one hand quite general, in that they are not too bound to the specific form of the initial distribution $f_0$, but on the other hand, as already observed from the numerical simulations of the binary interaction model, cf. Section~\ref{sect:spat.homog.Boltz_numsim}, only sufficient. In the case of the mean-field interaction model, we prove that particular initial distributions might produce alignment also in cases not explicitly covered by Theorem~\ref{theo:align.mean_field}.

\begin{corollary}
Let $f_0=\rho\delta_\alpha\in\cM^\rho_+(I)$, $\alpha\in I$, and let $f(t)\in\cM^\rho_+(I)$ be the corresponding solution to~\eqref{eq:quasinv_limit.strong} at time $t>0$. Then
$$ \lim_{t\to+\infty}W_1(f(t),\,\rho\delta_{\alpha_d})=0 $$
for all $\rho\in [0,\,1]$ and all $\alpha_c\in [-\pi,\,\pi)$.
\end{corollary}
\begin{proof}
The equation of the characteristics for the given $f_0$ is
$$ \dot{\gamma}_t(\theta)=-\rho(\gamma_t(\theta)-\alpha_d)+a(\rho)\rho\cG(\abs{\gamma_t(\theta)-\gamma_t(\alpha)})(\gamma_t(\theta)-\alpha_d+\alpha_c). $$
In particular, for $\theta=\alpha$ we obtain
$$ \dot{\gamma}_t(\alpha)=-\rho(\gamma_t(\alpha)-\alpha_d), $$
which, together with $\gamma_0(\alpha)=\alpha$, gives $\gamma_t(\alpha)=\alpha_d+(\alpha-\alpha_d)e^{-\rho t}$. Using this in the calculations of the proof of Theorem~\ref{theo:align.mean_field} yields finally
$$ W_1(f(t),\,\rho\delta_{\alpha_d})\leq\rho\int_I\abs{\gamma_t(\theta)-\alpha_d}f_0(\theta)\,d\theta=\rho^2\abs{\gamma_t(\alpha)-\alpha_d}
	=\rho^2\abs{\alpha-\alpha_d}e^{-\rho t}, $$
whence the thesis follows.
\end{proof}

We conclude that the extra microscopic information brought by the kinetic distribution function $f_0$ may be essential to get a full understanding of the alignment dynamics in particular cases.

\subsection{Numerical simulations}
\label{sect:hydrodynamic_numsim}
Next we illustrate the time-asymptotic evolution of the solution to~\eqref{eq:quasinv_limit.strong} by various numerical experiments. To discretise the equation we use the semi-Lagrangian scheme described in Appendix~\ref{app:semi-lagrangian}, which has good stability and accuracy properties even in the case of non-local fluxes.

\begin{figure}[!t]
\begin{center}
\includegraphics[width=0.9\textwidth]{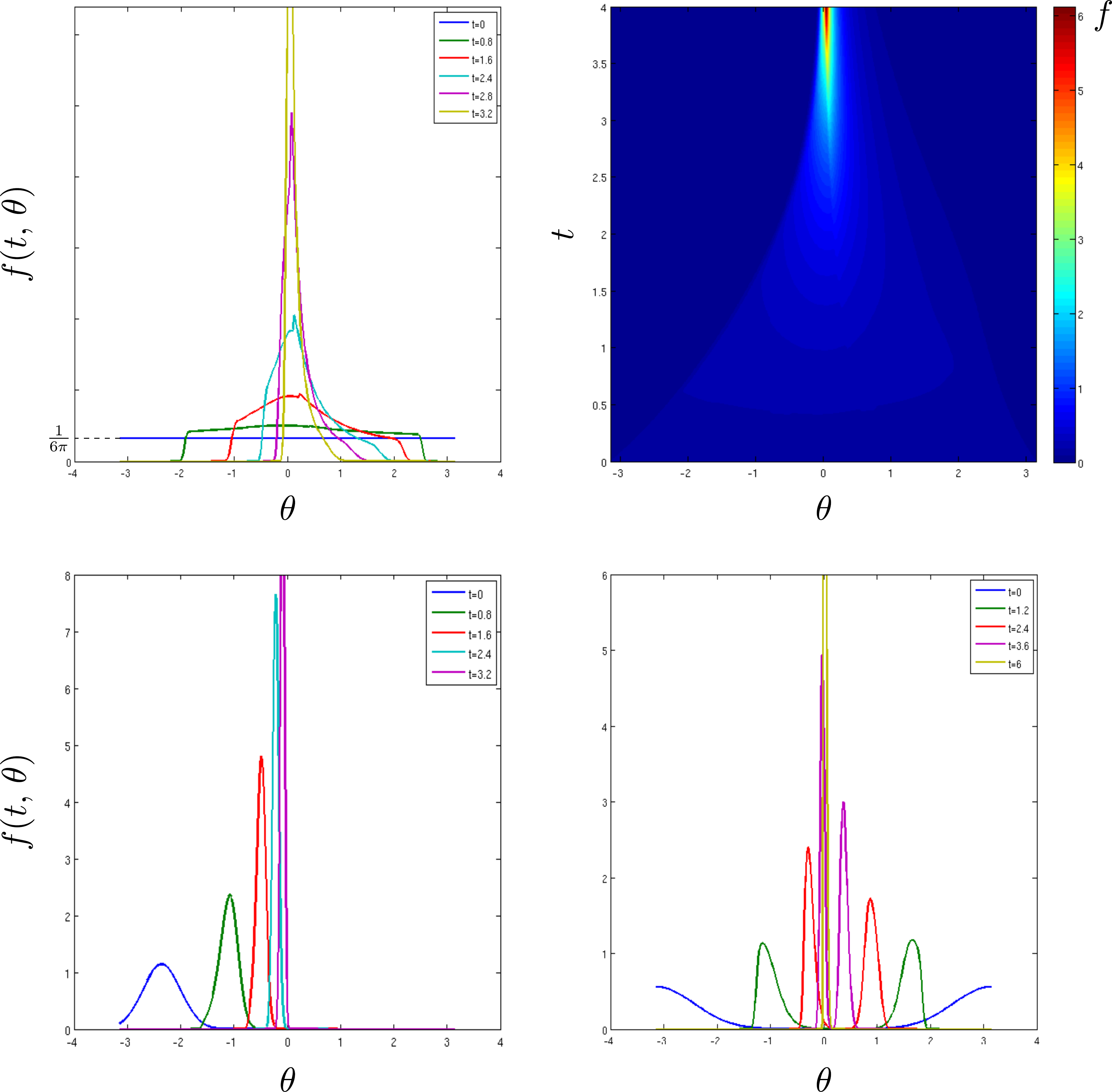}
\caption{\textbf{Top left}: snapshots in time of the solution to~\eqref{eq:quasinv_limit.strong} with initial condition~\eqref{eq:f0.uniform}. \textbf{Top right}: the same solution plotted in the $\theta t$-plane. \textbf{Bottom}: snapshots in time of the solution to~\eqref{eq:quasinv_limit.strong} with initial condition~\eqref{eq:f0.Gaussian} for two different choices of $\hat{\theta}$, $\sigma^2$.}
\label{fig:align_meanfield}
\end{center}
\end{figure}

We consider the uniform initial distribution~\eqref{eq:f0.uniform} with density $\rho=\frac{1}{3}$ and we take $\alpha_d=0$, $a(\rho)=\rho$, $\alpha_c=\frac{\pi}{3}$. In this case, consensus towards $\alpha_d$ is asymptotically expected because, as shown in the left panel of Figure~\ref{fig:comp_conditions}, these parameters fall in the range of unconditional alignment. The corresponding time evolution of the kinetic distribution function $f$ is displayed in the top row of Figure~\ref{fig:align_meanfield}.

Changing the initial condition, for instance using the function~\eqref{eq:f0.Gaussian} with various choices of the parameters $\hat{\theta}$, $\sigma^2$, does not affect the asymptotic behaviour of the system, cf. the bottom row of Figure~\ref{fig:align_meanfield}. This is indeed consistent with the theoretical prediction of Theorem~\ref{theo:align.mean_field}.

Similarly to the numerical tests presented in Section~\ref{sect:spat.homog.Boltz_numsim}, we also perform a parametric study of the trends of the mean-field interaction model for large times by computing a Finite-Element-type approximation, say $\thetaFEM$, of the quantity $\bar{\theta}(t)$. Specifically, we compute  the integral in~\eqref{eq:bartheta} by the rectangle method over a grid of points $\{\theta_i\}_{i=1}^{M}\subset I$ with step $\Delta\theta:=\frac{2\pi}{M+1}$, see Appendix~\ref{app:semi-lagrangian} for further details:
$$ \thetaFEM(t):=\frac{1}{\rho}\sum_{i=1}^M\abs{\theta_i-\alpha_d}f(t,\,\theta_i)\Delta\theta. $$

\begin{figure}[!t]
\begin{center}
\includegraphics[width=0.8\textwidth]{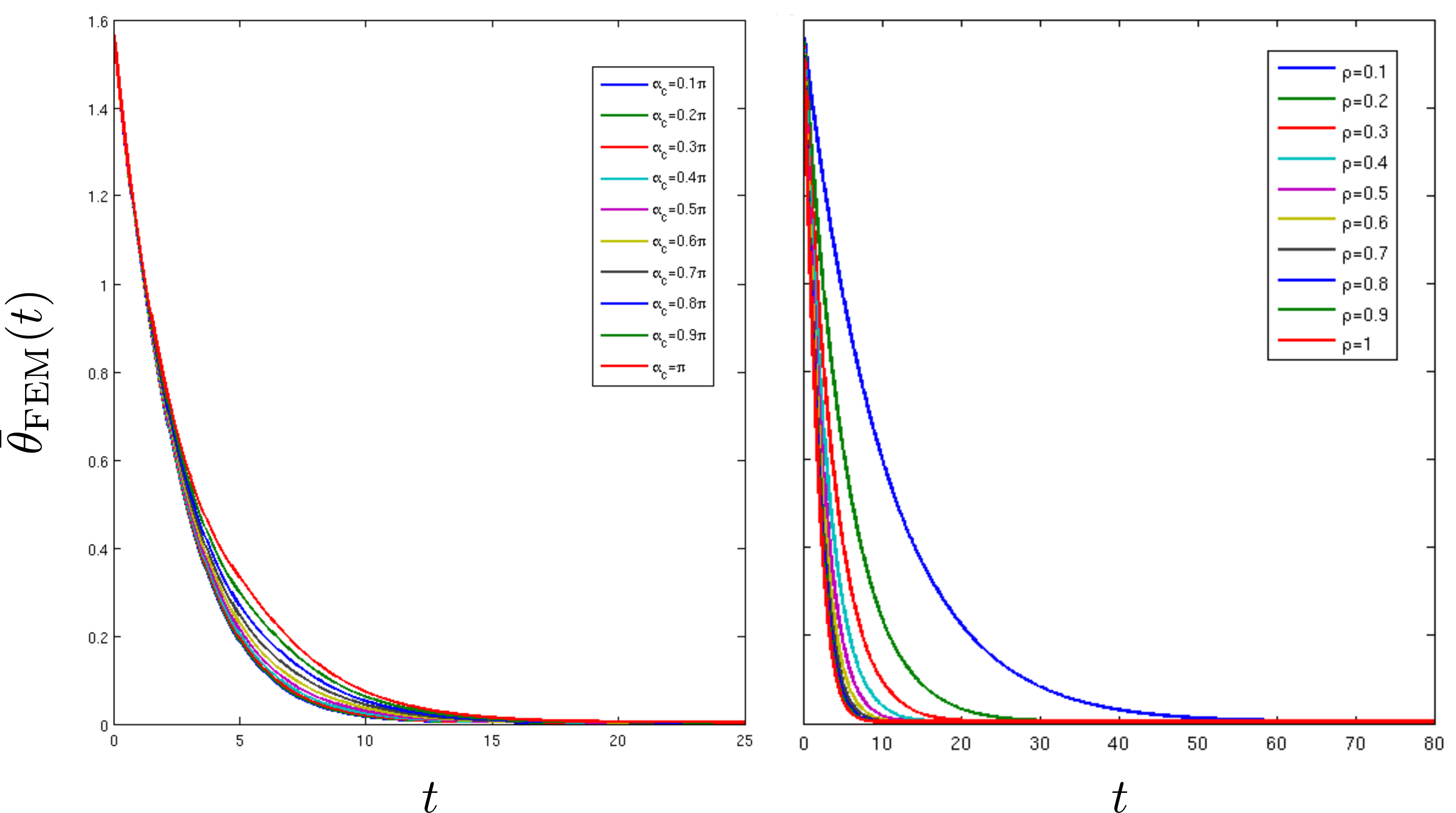}
\caption{The mapping $t\mapsto\thetaFEM(t)$ for: (\textbf{left}) $\rho=\frac{1}{2}$ and different choices of $\alpha_c\in (0,\,\pi]$; (\textbf{right}) $\alpha_c=\frac{\pi}{5}$ and different choices of $\rho\in (0,\,1]$. Alignment is always observed.}
\label{fig:thetabarFEM.parameters}
\end{center}
\end{figure}

In Figure~\ref{fig:thetabarFEM.parameters} we illustrate the evolution of the mapping $t\mapsto\thetaFEM(t)$ for:
\begin{enumerate}[(i)]
\item $\rho=\frac{1}{2}$ and setting $\alpha_c=\frac{k}{10}\pi$ for $k=1,\,\dots,\,10$, cf. Figure~\ref{fig:thetabarFEM.parameters} (left);
\item $\alpha_c=\frac{\pi}{5}$ and setting $\rho=\frac{k}{10}$ for $k=1,\,\dots,\,10$, cf. Figure~\ref{fig:thetabarFEM.parameters} (right).
\end{enumerate}
Case (i) is completely not covered by condition~\eqref{eq:cond.align.mean_field}; case (ii), instead, is covered for $\rho<\frac{5}{11}\approx 0.45$. In both cases, however, we do not find any numerical evidence that the alignment fails for some choices of the parameters. As already discussed, this is not in contrast with condition~\eqref{eq:cond.align.mean_field}, which indeed is only sufficient.

\section{The spatially inhomogeneous problem}
\label{sect:spat.inhom}
In the spatially inhomogeneous case one considers a generic distribution in space of the pedestrians, which can vary from point to point $x\in\R^2$ and over time. Therefore the kinetic distribution function fully depends on both state variables $(x,\,\theta)$, which implies that the density of the crowd:
$$ \rho(t,\,x):=\int_If(t,\,x,\,\theta)\,d\theta $$
is no longer a constant parameter of the model. However, from~\eqref{eq:Boltz}-\eqref{eq:Q} we see that the total mass of the system, namely the quantity $\int_I\int_{\R^2}f(t,\,x,\,\theta)\,dx\,d\theta=\int_{\R^2}\rho(t,\,x)\,dx$, is conserved in time because $\int_I\int_{\R^2}Q(f,\,f)(t,\,x,\,\theta)\,dx\,d\theta=0$ for all $t>0$. In the following we consider the case $\int_I\int_{\R^2}f(t,\,x,\,\theta)\,dx\,d\theta=1$, hence $f$ can be understood as a probability density.

In this section we explore computationally the space inhomogeneous problem by means of the binary collision model, cf.~\eqref{eq:Boltz}-\eqref{eq:Q} or~\eqref{eq:Boltz.weak}, which we simulate using a Nanbu-like Monte Carlo particle method, see Appendix~\ref{app:nanbu} for details. We show that the microscopic model of walking behaviour discussed in Sections~\ref{sect:micro.coll},~\ref{sect:kin_mod}, though much simplified with respect to other sophisticated models such as e.g.~\cite{degond2013KRM}, is nonetheless able to reproduce various emergent macroscopic patterns while being more realistic from the phenomenological point of view than the interaction models based on position-dependent repulsion potentials.

We calculate the collision probability using~\eqref{eq:prob.coll}, where the time to collision $t(x,\,y,\,\theta_\ast,\,\varphi_\ast)$ is given by~\eqref{eq:tij.polar} when the interacting pairs $(x,\,\theta_\ast)$, $(y,\,\varphi_\ast)$ are such that $\abs{x-y}>\gamma$ and $\Delta\geq 0$ in~\eqref{eq:Delta.polar}. Otherwise we set the time to collision to either $0$ or $O(10^5)$.

In the case studies addressed here we consider a two-dimensional bounded spatial domain, specifically the square $\Omega:=[-\frac{L}{2},\,\frac{L}{2}]\subset\R^2$ with edge length $L>0$ and periodic boundary conditions. Table~\ref{tab:parameters} states all the relevant parameters used in the simulations.

\begin{table}[!t]
\caption{Parameters used in the numerical simulations of Section~\ref{sect:spat.inhom}.}
\begin{tabular}{l|cccccccc}
	& $L$ & $\alpha_d$ & $\alpha_d^1$ & $\alpha_d^2$ & $\alpha_c$ & $\tau$ & $\gamma$ & $\eta$ \\
	\hline\hline
	Test 1 & $10$ & $\pi$ & $-$ & $-$ & $\frac{\pi}{4}$ & $1$ & $\frac{1}{2}$ & $-$ \\[2mm]
	Test 2 & $10$ & $-$ & $0$ & $-\pi$ & $\frac{\pi}{4}$ & $1$ & $\frac{1}{2}$ & $\frac{1}{5}$ \\[2mm]
	Test 3 & $10$ & $-$ & $0$ & $\frac{\pi}{2}$ & $\frac{\pi}{4}$ & $10$ & $\frac{2}{5}$ & $\frac{1}{5}$ \\[1mm]
	\hline
\end{tabular}
\label{tab:parameters}
\end{table}

\subsection{Test 1 -- Alignment}
\begin{figure}[!t]
\begin{center}
\includegraphics[width=\textwidth]{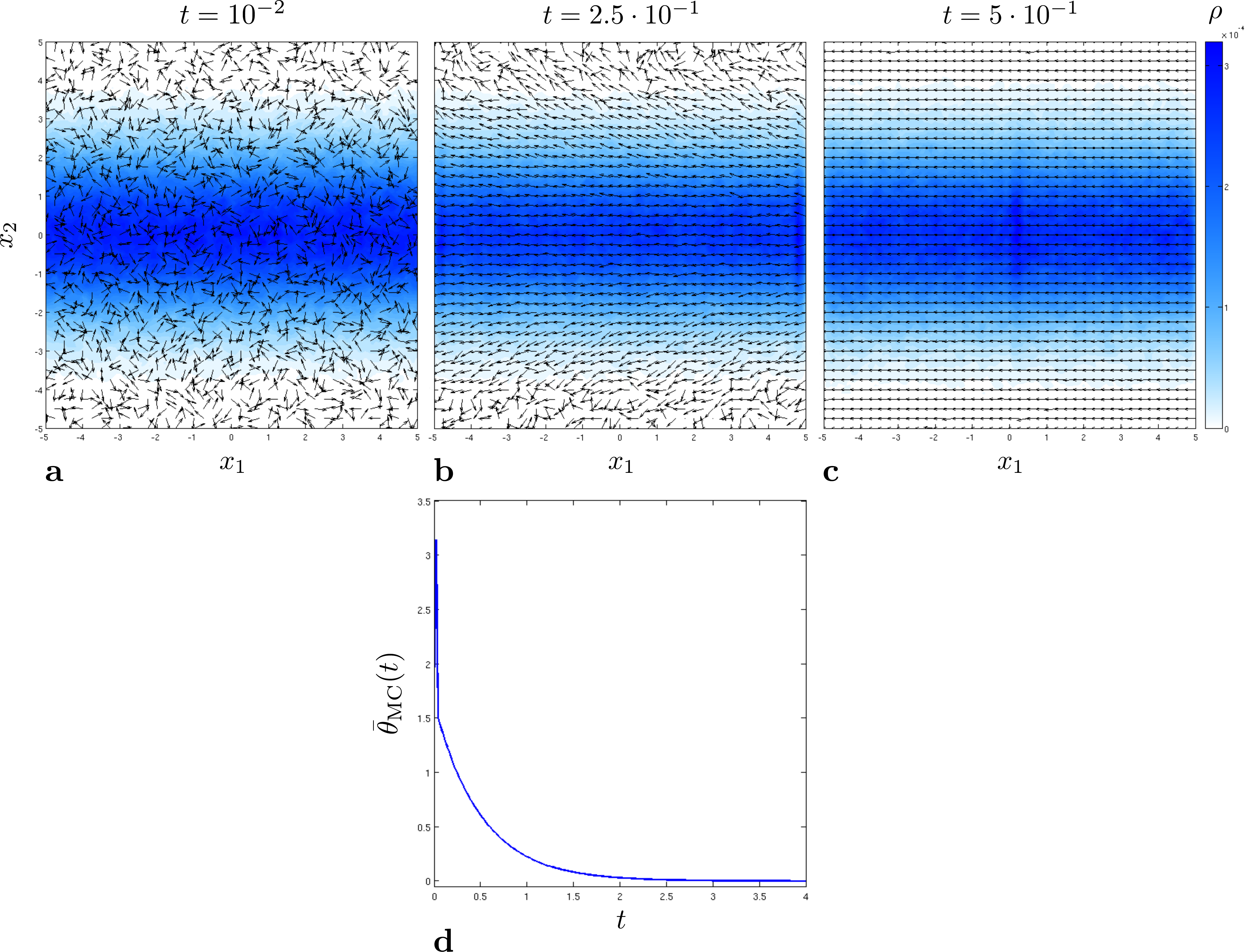}
\caption{Test 1 -- Alignment in the spatially inhomogeneous case and time trend of $\thetaMC$ over the whole space domain $\Omega$.}
\label{fig:align.onepop.spatinhomog}
\end{center}
\end{figure} 

In the first test we consider a similar situation as for the space homogeneous model in Section~\ref{sect:spat.homog}. We choose an initial condition of the form
$$ f_0(x,\,\theta)=\frac{1}{2\pi L}g(x_2) \quad \text{with} \quad
	g(x_2)=\frac{1}{\sqrt{2\pi}}\sum_{n=-\infty}^{+\infty}e^{-\frac{(x_2+nL)^2}{2}},\quad x_2\in\left[-\frac{L}{2},\,\frac{L}{2}\right), $$
such that $\int_I\int_\Omega f_0(x,\,\theta)\,dx\,d\theta=1$, which corresponds to a crowd concentrated mainly in a horizontal middle stripe of $\Omega$ with walking direction uniformly distributed in $I$, see Figure~\ref{fig:align.onepop.spatinhomog}a.

Like in the space homogeneous case, the pedestrians tend to align group-wise to the desired direction in finite time (cf. Figure~\ref{fig:align.onepop.spatinhomog}c). However, the simulation shows that such a consensus is faster where the density is higher (cf. Figure~\ref{fig:align.onepop.spatinhomog}b). This can indeed be inferred also from the results of the space homogeneous case, cf. the central panel of Figure~\ref{fig:thetabarMC.parameters}. Moreover, we notice that during the transition from the initial condition to the consensus the individuals tend to move from high to low density areas due to sidestepping for collision avoidance (cf. the upper and lower white stripes in Figure~\ref{fig:align.onepop.spatinhomog}b). Such an effect appears to be top/bottom-symmetric despite the fact that the deviation angle $\alpha_c$ is constant across the domain (in this case $\alpha_c>0$, cf. Table~\ref{tab:parameters}, corresponding to leftwards stepping).

This may explain the early non-monotone trend of the mapping $t\mapsto\thetaMC(t)$, cf. Figure~\ref{fig:align.onepop.spatinhomog}d, which in the present case is computed as a Monte Carlo approximation of
$$ \bar\theta(t):=\int_I\int_{\Omega}\abs{\theta-\alpha_d}f(t,\,x,\,\theta)\,dx\,d\theta. $$
In particular, the steep initial rise might be due to the fact that most individuals in the top area of the domain have to span anticlockwise all the walking directions in the interval $I$ before finding one free from collisions with neighbouring pedestrians. The subsequent exponential-like decay of $\thetaMC$ is instead consistent with the observations made in the spatially homogeneous case.

\subsection{Test 2 -- Counterflow}

A very popular problem in crowd dynamics is the counterflow, i.e. the case of two groups walking towards each other. In order to simulate this situation we consider an extension of model~\eqref{eq:Boltz} to two groups of pedestrians, each described by its own distribution function $f^p=f^p(t,\,x,\,\theta)$ which satisfies the equation
\begin{equation}
	\partial_tf^p+v\cdot\nabla_x{f^p}=Q(f^p,\,f^p)+Q(f^p,\,f^q), \qquad p,\,q\in\{1,\,2\},\ p\neq q.
	\label{eq:Boltz.twopop}
\end{equation}
The additional collisional term $Q(f^p,\,f^q)$ at the right-hand side takes into account the interactions of either group with the opposite one.

We assume that the two groups differ only in the desired directions, which we take to be oriented rightwards and leftwards, respectively. Hence each collisional term in~\eqref{eq:Boltz.twopop} implements an interaction rule of the form~\eqref{eq:int.rules} but with $\alpha_d$ replaced by $\alpha_d^p$, cf. Table~\ref{tab:parameters}. We consider for both groups an identical initial distribution
\begin{equation}
	f^p_0(x,\,\theta)=\frac{1}{2\pi\eta L^2}\chi_{[-\eta\frac{L}{2},\,\eta\frac{L}{2}]}(x_2), \qquad p=1,\,2,
	\label{eq:fp0}
\end{equation}
which is uniform in $\theta\in I$ and in $x\in\left[-\frac{L}{2},\,\frac{L}{2}\right]\times \left[-\eta\frac{L}{2},\,\eta\frac{L}{2}\right]\subset\Omega$, $0<\eta\leq 1$, and is such that $\int_I\int_{\Omega}f^p_0(x,\,\theta)\,dx\,d\theta=1$. Hence the two crowds share initially the same area of the domain and are well mixed therein (Figure~\ref{fig:twopop_lanes}a).

\begin{figure}[!t]
\begin{center}
\includegraphics[width=\textwidth]{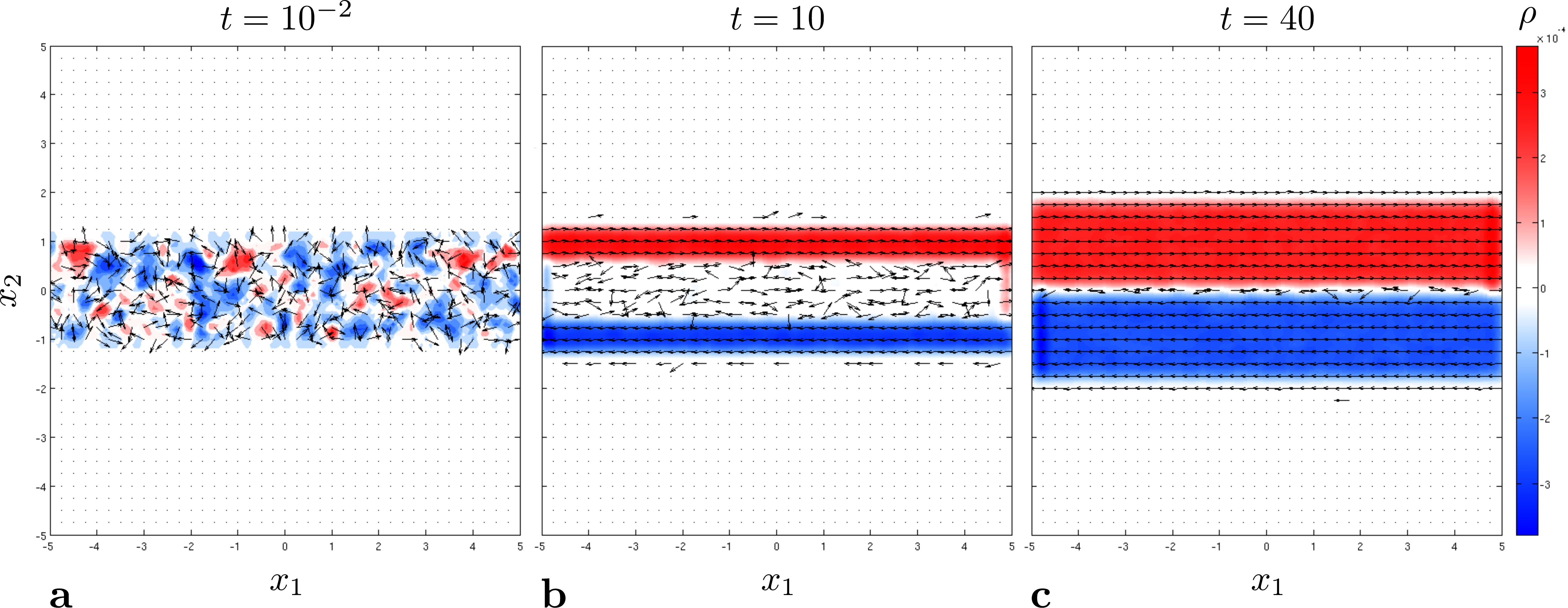}
\caption{Test 2 -- Lane formation in pedestrian counterflow. For pictorial purposes, in order to represent the two crowds in the same picture, the blue density is plotted on a \emph{negative} scale. In white areas where the velocity is non-zero the two densities take nearly the same values.}
\label{fig:twopop_lanes}
\end{center}
\end{figure}

The simulation shows that the two groups fully segregate, see Figure~\ref{fig:twopop_lanes}. A similar behaviour has been ovserved in the literature for different types of models, cf. e.g.~\cite{burger2016SIMA,cristiani2010CHAPTER,cristiani2014BOOK,moussaid2010PLOS}. Furthermore, within either lane each group aligns to its desired direction. As a matter of fact, this is the optimal way for pedestrians to avoid collisions with others.

In the early stages of the segregation (Figure~\ref{fig:twopop_lanes}b) the lanes are thinner and well separated by an area in which pedestrians with opposite desired directions are still mixed. In the long run, when either group aligns to its desired direction (Figure~\ref{fig:twopop_lanes}c), the lanes become wider and the groups clearly segregate.

\subsection{Test 3 -- Crossing flows}

We now consider two groups of pedestrians having perpendicular desired directions. As initial conditions we prescribe the distribution functions
$$ f^p_0(x,\,\theta)=\frac{1}{2\pi\eta L^2}\chi_{[-\eta\frac{L}{2},\,\eta\frac{L}{2}]}(x_q), \qquad p,\,q=1,\,2,\ p\neq q, $$
which are analogous to~\eqref{eq:fp0} but for the fact that here the smaller edge of the initial stripe is in the $x_q$-direction. Hence the group $p=1$ is distributed horizontally and walks rightwards while the group $p=2$ is distributed vertically and walks upwards, see Figure~\ref{fig:twopop_stripes}a and cf. Table~\ref{tab:parameters}.

The snapshots in the top row of Figure~\ref{fig:twopop_stripes} show that the pedestrians anticipate the interactions by starting to step aside slightly before they reach the area in which the two streams cross. As a result, the macroscopic flows deviate from their initial orthogonal directions thereby giving rise to quite realistic patterns. 

\begin{figure}[!t]
\begin{center}
\includegraphics[width=\textwidth]{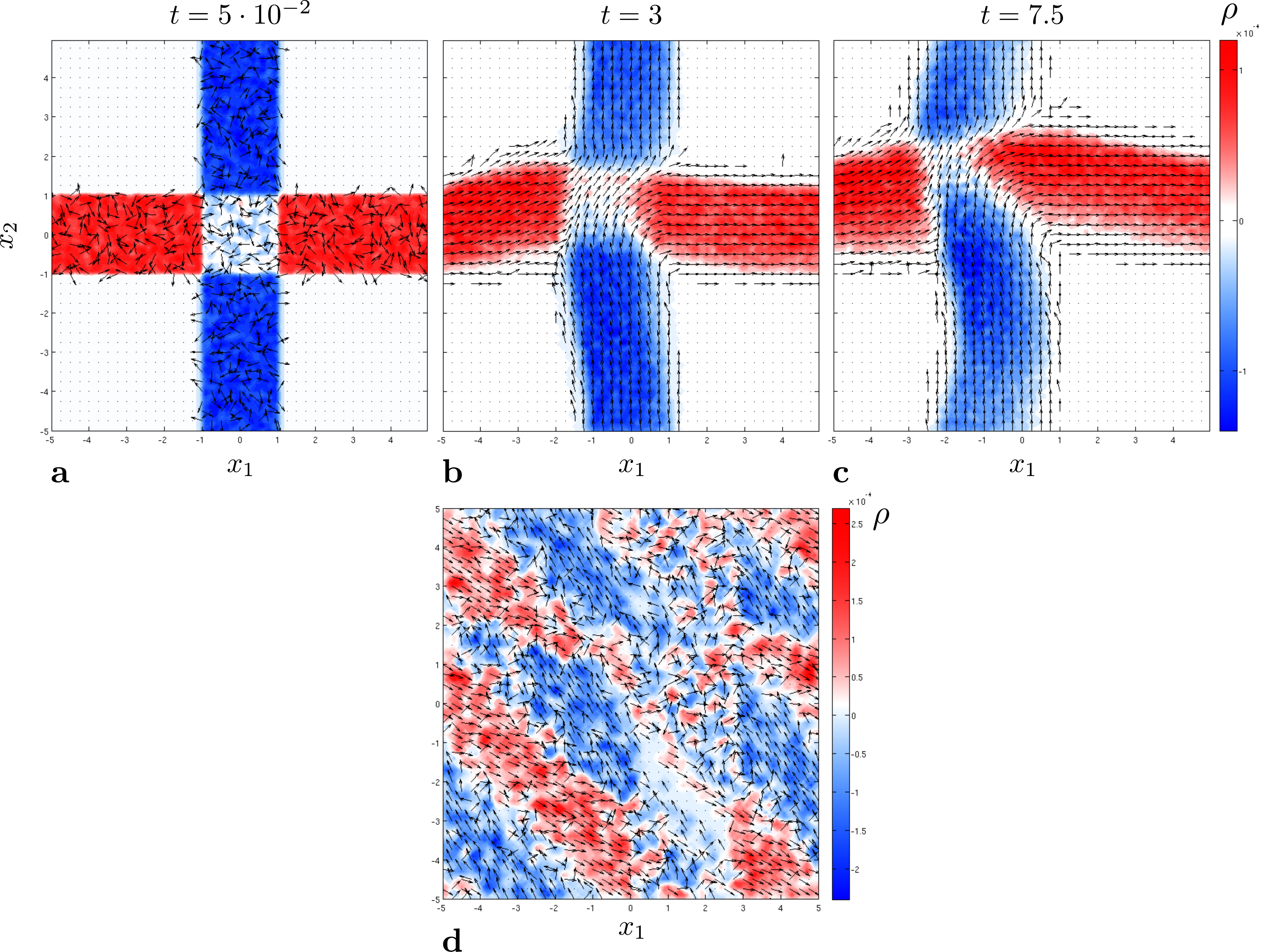}
\caption{Test 3 -- Crossing flows and stripe formation.}
\label{fig:twopop_stripes}
\end{center}
\end{figure}

In order to better investigate the dynamics in the crossing area, we consider at last two groups of pedestrians, with the same perpendicular desired directions as before, that are now uniformly distributed in the whole spatial domain $\Omega$. As Figure~\ref{fig:twopop_stripes}d shows, when they begin to interact we observe the emergence of segregation with \emph{stripe formation}, a phenomenon well-documented in the experimental literature~\cite{helbing2005TS,schadschneider2011NHM,zhang2014PHYSA}. The slope of the stripes corresponds to an angle between $\alpha_d^1$ and $\alpha_d^2$, which in the present case seems to be close to $\frac{\alpha_d^1+\alpha_d^2}{2}=\frac{\pi}{4}$. However, further analytical study of the model is needed to possibly confirm this guess.

\appendix
\section{Numerical tools}
\label{app:num.tools}
In this appendix we provide more details about the numerical schemes that we used for simulating the Boltzmann-type collisional model (both homogeneous and inhomogeneous in space) and the hydrodynamic mean-field model deduced from the former in the quasi-invariant direction limit.

\subsection{Nanbu-like Monte Carlo algorithm}
\label{app:nanbu}
The Nanbu algorithm is a particle method belonging to the family of the Monte Carlo numerical methods for the approximate solution of collisional kinetic equations. We used it for producing the simulations presented in Sections~\ref{sect:spat.homog.Boltz_numsim} and~\ref{sect:spat.inhom}.

Here we follow the approach presented in~\cite{albi2013MMS,pareschi2013BOOK}, see also~\cite{bobylev2000PRE}, to approximate the spatially inhomogeneous Boltzmann-type equation~\eqref{eq:Boltz} with microscopic states $(x,\,\theta)\in\mathbb{R}^3$ and $v=V_0(\cos{\theta}\mathbf{i}+\sin{\theta}\mathbf{j})$. The reader can easily adapt this description to the spatially homogeneous case, when the transport term $v\cdot\nabla_x{f}$ drops.

The algorithm is based on a splitting procedure: a first collisional step is followed by a transport step according to
\begin{subequations}
\begin{align}
	& \partial_t f(t,\,x,\,\theta)=Q(f,\,f)(t,\,x,\,\theta) \label{eq:MC.coll} \\[2mm]
	& \partial_t f(t,\,x,\,\,\theta)+v\cdot\nabla_x{f}(t,\,x,\,\theta)=0. \label{eq:MC.transp}
\end{align}
\end{subequations}

If we assume that, up to normalisation, the initial datum $f_0=f_0(x,\,\theta)$ is a probability density, i.e. $\int_I\int_{\R^2}f_0(x,\,\theta)\,dx\,d\theta=1$, then, since the total mass is conserved, $f(t,\,x,\,\theta)$ is in turn a probability density in $(x,\,\theta)$ for all $t>0$. Hence we can rewrite the operator $Q$ involved in the collisional step as, cf.~\eqref{eq:Q},
$$ Q(f,\,f)(t,\,x,\,\theta)=Q^+(f,\,f)(t,\,x,\,\theta)-f(t,\,x,\,\theta), $$
where
$$ Q^+(f,\,f)(t,\,x,\,\theta):=\int_I\int_{\R^2}\frac{1}{\abs{J}}f(t,\,x,\,\theta_\ast)f(t,\,y,\,\varphi_\ast)\,dy\,d\varphi $$
is the so-called \emph{gain operator}, which implements the interaction rule. Still because of mass conservation, or equivalently $\int_I\int_{\R^2}Q(f,\,f)(t,\,x,\,\theta)\,dx\,d\theta=0$, also the gain operator $Q^+(f,\,f)(t,\,x,\,\theta)$ is a probability density in $(x,\,\theta)$ for all $t>0$. 

\begin{algorithm}[!t]
	\caption{Nanbu-like algorithm for~\eqref{eq:Boltz}}
	\begin{algorithmic}[1]
		\STATE fix $N>1$, $\Delta{t}\in (0,\,1]$
		\STATE define the random variable $\xi\sim\text{Bernoulli}(\Delta{t})$
		\FOR{$n=0,\,1,\,2,\,\dots$}
			\STATE sample $N$ particles with states $\{(X_i^n,\,\Theta_i^n)\}_{i=1}^{N}$ from the distribution $f^n$
			\FOR{$i=1$ \TO $N$}
				\STATE sample a value of $\xi\in\{0,\,1\}$
				\IF{$\xi=0$}
					\STATE set $\Theta_i^{n+1}=\Theta_i^n$
				\ELSIF{$\xi=1$}
					\STATE sample uniformly an index $j\in\{1,\,\dots,\,N\}\setminus\{i\}$
					\STATE set $\Theta_i^{n+1}=\scrC(X_i^n,\,X_j^n,\,\Theta_i^n,\,\Theta_j^n)\mod 2\pi$, cf.~\eqref{eq:int.rules}
				\ENDIF
				\STATE set $X_i^{n+1}=X_i^n+V_0\left(\cos{\Theta_i^n}\mathbf{i}+\sin{\Theta_i^n}\mathbf{j}\right)\Delta{t}$
			\ENDFOR
			\STATE reconstruct (an approximation of) $f^{n+1}$ from the samples $\{(X_i^{n+1},\,\Theta_i^{n+1})\}_{i=1}^{N}$
		\ENDFOR
	\end{algorithmic}
	\label{alg:nanbu}
\end{algorithm}

Equation~\eqref{eq:MC.coll} is discretised in time by the forwards Euler scheme to get
\begin{equation}
	f^{n+1}(x,\,\theta)=(1-\Delta{t})f^n(x,\,\theta)+\Delta{t}Q^+(f^n,\,f^n)(x,\,\theta),
	\label{eq:MC.fwdEul}
\end{equation}
where $\Delta{t}>0$ is the time step and $f^n$ is an approximation of $f$ at time $t_n:=n\Delta{t}$, $n=0,\,1,\,2,\,\dots$.

Under the constraint $\Delta{t}\leq 1$, from~\eqref{eq:MC.fwdEul} we see that $f^{n+1}$ is a convex linear combination of $f^n$ and $Q^+(f^n,\,f^n)$, hence it is in turn a probability density. This is the key to give~\eqref{eq:MC.fwdEul} the following probabilistic interpretation in terms of the underlying microscopic particle system: to obtain samples distributed according to $f^{n+1}$ we have to sample either from $f^n$, with probability $1-\Delta{t}$, or from $Q^+(f^n,\,f^n)$, with probability $\Delta{t}$. As a matter of fact, in~\eqref{eq:MC.fwdEul} $f^n$ is related to the event that no binary interactions take place during the time $\Delta{t}$, so that the kinetic distribution does not change in the time step $n\to n+1$, while $Q^+(f^n,\,f^n)$ is related to the event that an interaction occurs.

Therefore, after sampling from $f^n$ a certain number $N$ of particles with states $\{(X_i^n,\,\Theta_i^n)\}_{i=1}^{N}$, with probability $1-\Delta{t}$ we set $\Theta_i^{n+1}=\Theta_i^n$, i.e. we leave the direction of the $i$th particle unchanged. Conversely, with probability $\Delta{t}$ we:
\begin{enumerate}[(i)]
\item select uniformly a second particle $(X_j^n,\,\Theta_j^n)$, $j\in\{1,\,\dots,\,N\}$, with $j\ne i$;
\item update $\Theta_i^n$ to $\Theta_i^{n+1}$ according to the binary interaction rule~\eqref{eq:int.rules} applied to the $i$th and $j$th particles.
\end{enumerate}

In the subsequent transport step we update the positions of the sampled particles. Owing to~\eqref{eq:MC.transp}, we have:
$$ X_i^{n+1}=X_i^n+V_0\left(\cos{\Theta_i^n}\mathbf{i}+\sin{\Theta_i^n}\mathbf{j}\right)\Delta{t}, \qquad i=1,\,\dots,\,N. $$

The whole method is summarised in Algorithm~\ref{alg:nanbu} in the form of a pseudo-code. 

\subsection{Semi-Lagrangian scheme for conservation laws with non-local flux}
\label{app:semi-lagrangian}
In recent works, see e.g.~\cite{crouseilles2010JCP,huang2016JCP}, semi-Lagrangian (SL) schemes have been proved efficient in approximating conservation laws. Since these numerical methods are not classical and need some adaptations to our case, we report here a short description of their derivation for the reader's convenience. In particular, following~\cite{carlini2016IFAC,carlini2016DGA}, see also~\cite{piccoli2013AAM,tosin2011NHM}, we describe a SL scheme for the numerical approximation of the following problem, cf. Section~\ref{sect:hydrodynamic_numsim}: 
 \begin{equation}
	\begin{cases}
		\partial_tf+\partial_\theta(H[f]f)=0 & (\theta,\,t)\in I\times (0,\,T] \\[1mm]
		f(0,\,\theta)=f_0(\theta) & \theta\in I,
	\end{cases}
	\label{eq:cons.law}
\end{equation}
where $T>0$ is a final time, $I=[-\pi+\alpha_d,\,\pi+\alpha_d)$ is the interval introduced in Remark~\ref{rem:alpha_d} with periodic boundary conditions and $H[f]=H[f](t,\,\theta):(0,\,T]\times I\to\R$ is a sufficiently smooth and bounded velocity, which possibly depends on the unknown $f$ in a non-local manner.

To derive the scheme, we begin by fixing a time step $\Delta{t}>0$ through which we define the discrete time instants $t_n:=n\Delta{t}$, $n=0,\,1,\,2,\,\dots$. Then we multiply the first equation in~\eqref{eq:cons.law} by a sufficiently smooth test function $\psi:I\to\R$ with $\psi(\pm\pi+\alpha_d)=0$ and integrate over $I\times [t_n,\,t_{n+1}]$ using integration by parts with respect to $\theta$:
$$ \int_I\psi(\theta)f(t_{n+1},\,\theta)\,d\theta=\int_I\psi(\theta)f(t_n,\,\theta)\,d\theta
	+\int_{t_n}^{t_{n+1}}\int_I\psi'(\theta)H[f](t,\,\theta)f(t,\,\theta)\,d\theta\,dt. $$
A first order approximation of the integral in time at the right-hand side gives
$$ \int_I\psi(\theta)f(t_{n+1},\,\theta)\,d\theta=
	\int_I\Bigl(\psi(\theta)+\psi'(\theta)H[f](t_n,\,\theta)\Delta{t}\Bigr)f(t_n,\,\theta)\,d\theta+o(\Delta{t}), $$
whence, in view of the expansion $\psi(\theta)+\psi'(\theta)H[f](t_n,\,\theta)\Delta{t}=\psi\bigl(\theta+H[f](t_n,\,\theta)\Delta{t}\bigr)+o(\Delta{t})$,
\begin{equation}
	\int_I\psi(\theta)f(t_{n+1},\,\theta)\,d\theta=
		\int_I\psi\left(\Phi^n[f](\theta)\right)f(t_n,\,\theta)\,d\theta+o(\Delta{t}),
	\label{eq:SL.pushfwd}
\end{equation}
where $\Phi^n[f](\theta):=\theta+H[f](t_n,\,\theta)\Delta{t}$ is the so-called \emph{(discrete-in-time) flow map}. In the case of~\eqref{eq:H[f]} it results
$$ \Phi^n[f](\theta)=\theta+\left[\rho(\alpha_d-\theta)
	+a(\rho)(\alpha_c-\alpha_d+\theta)\int_I\cG(\abs{\theta-\varphi})f(t_n,\,\varphi)\,d\varphi\right]\Delta{t}. $$

Next we introduce a grid in $I$ by means of the points $\theta_i:=\theta_0+i\Delta{\theta}$, $i=1,\,\dots,\,M$, where $\theta_0$ is any point in $I$ and $\Delta{\theta}:=\frac{2\pi}{M+1}$ is chosen in such a way that $\theta_{M+1}=\theta_0\mod{2\pi}$. We define the pairwise disjoint grid cells $E_i:=\left[\theta_i-\frac{\Delta{\theta}}{2},\,\theta_i+\frac{\Delta{\theta}}{2}\right)$ and the cell averages of the function $f$ at time $t_n$:
$$ f_i^n:=\frac{1}{\Delta{\theta}}\int_{E_i}f(t_n,\,\theta)\,d\theta, \qquad i=1,\,\dots,\,M, $$
which are such that $f_i^n\approx f(t_n,\,\theta_i)$ for $\Delta{\theta}$ sufficiently small. Over such a grid we consider a Finite-Element-type approximation of the integrals appearing in~\eqref{eq:SL.pushfwd} and, at the same time, we neglect the remainder $o(\Delta{t})$ while still enforcing the equality between the left and right-hand sides to get:
$$ \sum_{i=1}^M\psi(\theta_i)f_i^{n+1}=\sum_{i=1}^{M}\psi(\Phi^n[f](\theta_i))f_i^n, \qquad n=0,\,1,\,2,\,\dots $$

Choosing as test functions in this equation the $\mathbb{P}_1$-basis functions $\{\beta_i(\theta)\}_{i=1}^M$ associated with the nodes $\{\theta_i\}_{i=1}^M$:
$$ 	\beta_i(\theta):=\left(1-\dfrac{\abs{\theta-\theta_i}}{\Delta{\theta}}\right)\chi_{[\theta_{i-1},\,\theta_{i+1}]}(\theta),
	\qquad
	\beta_i(\theta_j)=
	\begin{cases}
		1 & \text{if } j=i \\
		0 & \text{otherwise,}
	\end{cases}
$$
we obtain the explicit-in-time scheme
$$ f_i^{n+1}=\sum_{j=1}^M\beta_i(\Phi^n[f](\theta_j))f_j^n, \qquad i=1,\,\dots,\,M,\ n=0,\,1,\,2,\,\dots $$
with initial data
$$ f_i^0:=f_0(\theta_i)\approx\frac{1}{\Delta{\theta}}\int_{E_i}f_0(\theta)\,d\theta, \qquad i=1,\,\dots,\,M, $$
where $f_0$ is the initial condition prescribed in~\eqref{eq:cons.law}. In case of~\eqref{eq:H[f]}, the value $\Phi^n[f](\theta_j)$ is approximated as
$$ \Phi^n[f](\theta_j)\approx\theta_j+\left[\rho(\alpha_d-\theta_j)
	+a(\rho)(\alpha_c-\alpha_d+\theta_j)\sum_{k=1}^M\cG(\abs{\theta_j-\theta_k})f_k^n\Delta{\theta}\right]\Delta{t}. $$

\section*{Acknowledgments}
A.F. and M.-T.W. have been supported by the New Frontiers Grant NST 0001 of the Austrian Academy of Sciences.

A.T. is member of GNFM (Gruppo Nazionale per la Fisica Matematica) of INdAM (Istituto Nazionale di Alta Matematica), Italy.

A.T. acknowledges that this work has been supported by a starting grant of the Compagnia di San Paolo (Turin, Italy).

\bibliographystyle{amsplain}
\bibliography{references}

\end{document}